\documentclass[12pt]{article}
\usepackage[margin=1.25in]{geometry}
\usepackage{amsmath,amssymb,amsthm}
\usepackage[authoryear]{natbib}
\usepackage{titletoc}
\usepackage[dvipsnames]{xcolor}
\usepackage{enumitem}
\usepackage{rotating}
\usepackage{lscape}
\usepackage{float}
\usepackage{graphicx}
\usepackage{subcaption}
\usepackage{booktabs}
\usepackage{caption}
\usepackage{makecell}
\usepackage{adjustbox}
\usepackage{multirow}

 \usepackage[labelsep=period]{caption}
\newtheorem{theorem}{Theorem}

\newtheorem{assumption}{Assumption}
\newtheorem{lemma}{Lemma}
\newtheorem{definition}{Definition}

\usepackage[plainpages=false,pdfpagelabels]{hyperref}

\definecolor{darkblue}{rgb}{0, 0, 0.5}
\hypersetup{
     colorlinks = true,
     citecolor = Maroon,
     urlcolor = darkblue,
     linkcolor = darkblue
}

\begin{document}

\title{Detecting $p$-hacking\thanks{We are grateful to Brendan Beare, Gregory Cox, Bulat Gafarov, Xinwei Ma, Ulrich M\"uller, Christoph Rothe, Yixiao Sun, the Editor (Guido Imbens), anonymous referees, seminar participants at the National University of Singapore, the University of Cambridge, the University of Illinois at Urbana-Champaign, the University of Mannheim, and conference participants at the California Econometrics Conference 2019, the CEME Conference for Young Econometricians 2019, and the 2019 SEA End-Of-Year Conference for valuable comments. K.W. is also affiliated with CESifo and ifo Institute. The usual disclaimer applies.}}

\author{Graham Elliott\thanks{Department of Economics, University of California, San Diego, 9500 Gilman Dr. La Jolla, CA 92093. Email: \url{grelliott@ucsd.edu}}  \quad \quad Nikolay Kudrin\thanks{Department  of Economics, University of California, San Diego,  9500 Gilman Dr.\ La Jolla, CA 92093. Email: \url{nkudrin@ucsd.edu}}  \quad \quad Kaspar W\"uthrich\thanks{Department  of Economics,  University of California, San Diego,  9500 Gilman Dr.\ La Jolla, CA 92093. Email: \url{kwuthrich@ucsd.edu}}}

\maketitle


\begin{abstract} 

We theoretically analyze the problem of testing for $p$-hacking based on distributions of $p$-values across multiple studies. We provide general results for when such distributions have testable restrictions (are non-increasing) under the null of no $p$-hacking. We find novel additional testable restrictions for $p$-values based on $t$-tests. Specifically, the shape of the power functions results in both complete monotonicity as well as bounds on the distribution of $p$-values. These testable restrictions result in more powerful tests for the null hypothesis of no $p$-hacking. When there is also publication bias, our tests are joint tests for $p$-hacking and publication bias. A reanalysis of two prominent datasets shows the usefulness of our new tests.  

\bigskip

\noindent \textbf{Keywords:} $p$-values, $p$-curve, complete monotonicity, publication bias

\end{abstract}

\newpage

\section{Introduction}
A researcher's ability to explore various ways of analyzing and manipulating data and then selectively report the ones that yield better-looking results, commonly referred to as \emph{$p$-hacking}, compromises the reliability of research and undermines the scientific credibility of reported results. 
Absent systematic replication studies or meta analyses, a popular approach for assessing the extent of $p$-hacking is to examine distributions of $p$-values across studies, referred to as \emph{$p$-curves} \citep{simonsohn2014p}; see Section 2 in \citet{christensen2018transparency} for a review.\footnote{Examples include: \citet{masicampo2012peculiar}, \citet{leggett2013life}, \citet{simonsohn2014p,simonsohn2015better}, \citet{head2015extent}, \citet{dewinter2015surge}, and \citet{snyder2018sniff}. Another strand of the literature uses the distribution of $t$-statistics to test for $p$-hacking \citep[e.g.,][]{gerber2008do,brodeur2016,brodeur2020methods,bruns2019reporting,vivalt2019specification}.}

We consider the problem of testing the \emph{null hypothesis of no $p$-hacking} against the \emph{alternative hypothesis of $p$-hacking} and provide theoretical foundations for developing tests for $p$-hacking. We characterize analytically under general assumptions the null set of distributions of $p$-values implied in the absence of $p$-hacking and provide general sufficient conditions under which, for any distribution of the true effects, the $p$-curve is non-increasing and continuous in the absence of $p$-hacking. These conditions are shown to hold for many, but not all popular approaches to testing for effects. 

For the leading case where $p$-curves are based on $t$-tests, we derive additional previously unknown testable restrictions. Specifically, the $p$-curves based on $t$-tests are completely monotone in the absence of $p$-hacking, and their magnitude and the magnitude of their derivatives are restricted by upper bounds. These restrictions are particularly useful when $p$-hacking fails to induce an increasing $p$-curve---for example when researchers engage in specification search across independent tests. In such cases tests based on non-increasingness have no power. 

Our theoretical results allow us to develop more powerful statistical tests for $p$-hacking, which we apply to two large datasets of $p$-values.
We find evidence for $p$-hacking in settings where the existing tests do not reject the null of no $p$-hacking.

When there is publication bias, our results characterize the $p$-curve under the null hypothesis of \emph{no $p$-hacking and no publication bias}. Our tests become joint tests for $p$-hacking and publication bias, complementing available methods for identifying  publication bias \citep[see, e.g.,][and the references therein]{andrews2019identification}.

\section{The p-curve based on general tests}
\label{sec:setup1}
Here we provide general sufficient conditions under which the $p$-curve is non-increasing under the null hypothesis of no $p$-hacking. These results are useful because tests for $p$-hacking often assume non-increasingness of the $p$-curve \citep[e.g.,][]{simonsohn2014p,simonsohn2015better,head2015extent}. 
This assumption has been justified through analytical and numerical examples, which rely on specific choices of tests and distributions of true effects being tested \citep[e.g.,][]{hung1997,simonsohn2014p,ulrich2018some}. However, such analyses are not sufficient for guaranteeing size control of statistical tests for $p$-hacking since the true effect distribution is never known. Instead, what is required for size control in a wide range of applications is a characterization of the shape of the $p$-curve for general tests and effect distributions.

\subsection{Setup}

Consider a test statistic $T$ that is distributed according to a distribution with cumulative distribution function (CDF) $F_{h}$, where $h$ indexes parameters of either the exact or asymptotic distribution of the test. We assume that the parameters $h$ only contain the parameters of interest. This is suitable for settings with large enough samples and asymptotically pivotal test statistics, which are prevalent in applied research.

Suppose researchers are testing the hypothesis
\begin{equation}
H_0: h \in \mathcal{H}_0 \qquad \text{against} \qquad  H_1: h \in \mathcal{H}_1  \label{eq:testing_problem},
\end{equation} 
where $ \mathcal{H}_0 \cap \mathcal{H}_1 =\emptyset$. Let $ \mathcal{H} = \mathcal{H}_0 \cup \mathcal{H}_1 $.
Denote as $F$ the CDF of the chosen null distribution from which critical values are determined.  
We assume that the test rejects for large values of the test statistic and denote the critical value for a level $p$ test as $cv(p)$. We will focus on settings with a continuous and strictly increasing $F$ (see Assumption \ref{ass:regularity} below) and set $cv(p)=F^{-1}(1-p)$. For any $h$, we denote by $\beta\left(p,h \right) = \Pr\left(T>cv(p)\mid h\right)=1-F_{h}\left(cv(p)\right)$  the rejection rate of a level $p$ test with parameters $h$. For $ h \in \mathcal{H}_1$, this is the power of the test, and we refer to $\beta(p,h)$ as the \emph{power function}. 

For the remainder of the paper, we focus on settings where the tests generating the $p$-values satisfy Assumption \ref{ass:regularity}.  This allows us to work with a well-defined density function and provide general results. 
\begin{assumption}[Regularity] \label{ass:regularity} $F$ and $F_h$ are twice continuously differentiable with uniformly bounded first and second derivatives $f, f', f_h$ and $f_h'$. $f(x)>0$ for all $x\in \{cv(p): p\in (0,1)\}$. For $h\in \mathcal{H}$, $\text{supp}(f) = \text{supp}(f_h)$.\footnote{For a function $\varphi$, we define $\text{supp}(\varphi)$ to be the closure of $\{x: \varphi(x)\ne0\}$.}
\end{assumption}
Assumption \ref{ass:regularity} holds for many tests with parametric $F$ and $F_h$, including $t$-tests and Wald-tests. A necessary condition for Assumption \ref{ass:regularity} is the absolute continuity of $F$ and $F_h$. This is not too restrictive since, in many cases, $F$ and $F_h$ are the asymptotic distributions of test statistics, which typically satisfy this condition. Further, in cases where the test statistics have a discrete distribution, size does not typically equal level, which could lead to $p$-curves that violate non-increasingness.

Consider the distribution of the $p$-values across studies, where we compute $p$-values from a distribution of $T$ given values of $h$, which themselves are drawn from a probability distribution $\Pi$. We refer to $\Pi$ as the \emph{distribution of true effects}. The CDF of the $p$-values is 
\begin{eqnarray}
G(p)=\int_\mathcal{H}\Pr\left(T>cv(p)\mid h\right)d\Pi(h)=\int_\mathcal{H}\beta\left(p,h\right) d\Pi(h).\label{eq: CDF p-values}
\end{eqnarray}
Under Assumption \ref{ass:regularity}, define the $p$-curve as follows.
\begin{definition}[$P$-curve] The density of the $p$-values, the $p$-curve, is defined as
\begin{equation*}
g(p):=\int_\mathcal{H}\frac{\partial \beta\left(p,h\right)}{\partial p} d\Pi(h).
\end{equation*}
\end{definition}
In Section \ref{sec:p_curve_absence}, we analyze the shape of $g$ for general tests and distributions $\Pi$.

\subsection{Properties of p-curves based on general tests}
\label{sec:p_curve_absence}
Here we derive conditions under which the $p$-curve is non-increasing in the absence of $p$-hacking for any distribution of true effects. We show that this property holds for most but not all popular statistical tests. 

Under Assumption \ref{ass:regularity}, the curvature of the $p$-curve follows from   
\begin{equation*}
g'(p):=\frac{d g(p)}{d p}=\int_\mathcal{H}\frac{\partial^2 \beta\left(p,h\right)}{\partial p^2} d\Pi(h).
\end{equation*}
The sign of $g'(p)$ is determined by the second derivative of the rejection probability, $\partial^2 \beta\left(p,h\right)/\partial p^2$. As we will show in the proof of Theorem \ref{thm:monotonicity} below, the following condition implies that $\partial^2 \beta\left(p,h\right)/\partial p^2$ is non-positive for all $h\in \mathcal{H}$. 
\begin{assumption}[Sufficient condition] \label{ass:sufficient_condition} For all $(x,h)\in \{cv(p): p\in (0,1)\}\times \mathcal{H}$, 
\[
f'_{h}(x)f(x)\ge f'(x)f_{h}(x).
\]
\end{assumption}

Assumption \ref{ass:sufficient_condition} is a restriction on how the power function changes when the critical value changes, which is governed by the shape of the density.
When $\mathcal{H}_0=\{0\}$ and $F=F_0$ (as, for example, for one-sided $t$-tests), Assumption \ref{ass:sufficient_condition} is of the form of a monotone likelihood ratio property, which relates the shape of the density of $T$ under the null to the shape of the density of $T$ under alternative $h$. The next lemma shows that this condition holds for many popular tests. Let $\Phi$ denote the CDF of the standard normal distribution.

\begin{lemma}\label{lem:verification_sufficient_condition} Assumption \ref{ass:sufficient_condition} holds when
\begin{enumerate}[label=(\roman*)]\setlength\itemsep{-1.5pt}
\item $F(x)=\Phi(x)$, $F_h(x)=\Phi(x-h)$, $\mathcal{H}_0=\{0\}$, $\mathcal{H}_1\subseteq (0,\infty)$ (e.g., similar one-sided $t$-test)
\item $F$ is the CDF of a half-normal distribution with scale parameter $1$, $F_h$ is the CDF of a folded normal distribution with location parameter $h$ and scale parameter $1$, $\mathcal{H}_0=\{0\}$, $\mathcal{H}_1\subseteq \mathbb{R}\backslash \{0\}$ (e.g., two-sided $t$-test)
\item $F$ is the CDF of a $\chi^2$ distribution with degrees of freedom $d>0$, $F_h$ is the CDF of a noncentral $\chi^2$ distribution with degrees of freedom $d>0$ and noncentrality parameter $h$, $\mathcal{H}_0=\{0\}$, $\mathcal{H}_1\subseteq (0, \infty)$ (e.g., Wald test\footnote{
For instance, let $\sqrt{N}(\hat{\theta}-\theta)\overset{a}\sim \mathcal{N}(0, V)$, where $\hat\theta$ is an estimator of $\theta$ based on $N$ observations and $V\in \mathbb{R}^{\dim(\theta)\times \dim(\theta)}$ is known (or can be consistently estimated). Consider the problem of testing $H_0: R\theta=r$ against $H_1: R\theta\ne r$, where $R\in \mathbb{R}^{q\times \dim(\theta)}$, $r\in\mathbb{R}^q$, and $rank(R)=q$. Set $T=N(R\hat\theta-r)'(RVR')^{-1}(R\hat\theta-r)$. This fits our framework with $d=q$ and  $h:=\lambda'(RVR')^{-1}\lambda$, where $\lambda:=\sqrt{N}(R\theta-r)$.})
\end{enumerate}
\end{lemma}

The following theorem shows that the $p$-curve is non-increasing and continuously differentiable under the maintained assumptions for any distribution of true effects.

\begin{theorem}[Testable restrictions for general tests]\label{thm:monotonicity} Under Assumptions \ref{ass:regularity}--\ref{ass:sufficient_condition}, $g$ is continuously differentiable and $g'(p)\le 0$ for $p\in (0,1)$. 
\end{theorem}

The result in Theorem \ref{thm:monotonicity} holds for many commonly-used statistical tests such that, in many empirically relevant settings, the $p$-curve will be non-increasing in the absence of $p$-hacking. To our knowledge, Theorem \ref{thm:monotonicity} provides the first general formal justification for the existing tests for $p$-hacking that exploit non-increasingness of the $p$-curve. Theorem \ref{thm:monotonicity} further motivates the use of density discontinuity tests as an alternative to tests based on non-increasingness of the $p$-curve. 

The results can be extended to settings with nuisance parameters. In such settings, $h$ contains both the parameters of interest, $h_1$, as well as additional nuisance parameters, $h_2$, such that $h=(h_1,h_2)$. Let $\mathcal{H}^1$ and $\mathcal{H}^2$ denote the supports of $h_1$ and $h_2$. Allow the null distribution to depend on $h_2$ with CDF $F_{h_2}$. 
The CDF of $p$-values becomes
\begin{eqnarray*}
G(p)=\int_{\mathcal{H}^1\times \mathcal{H}^2}\beta\left(p,h_1,h_2\right) d\Pi(h_1,h_2),
\end{eqnarray*}
where $\beta(p,h_1,h_2)=1-F_{h}\left(cv_{h_2}(p)\right)$ and $cv_{h_2}(p)=F_{h_2}^{-1}(1-p)$. 
The results of Theorem \ref{thm:monotonicity} extend to the $p$-curve generated from this distribution after changing the notation to include the dependence on $h_2$. For $h_2\in \mathcal{H}^2$, $F_{h_2}$, $f_{h_2}, f_{h_2}'$ have the same properties as $F$, $f, f'$ in Assumption \ref{ass:regularity}, and the assumptions on $F_{h}$, $f_{h}, f_{h}'$ hold for $h=(h_1,h_2)$. Assumption \ref{ass:sufficient_condition} becomes $f'_{h}(cv_{h_2}(p))f_{h_2}(cv_{h_2}(p))\ge f_{h_2}'(cv_{h_2}(p))f_{h}(cv_{h_2}(p))$ for $(h_1,h_2)\in \mathcal{H}^1\times \mathcal{H}^2$. The proof then follows directly from that of Theorem \ref{thm:monotonicity}. 

In applications, often only a part of the $p$-curve is examined. 
The $p$-curve over subintervals $\mathcal{I}\subset (0,1)$ is given by $g_{\mathcal{I}}(p)=g(p)/\int_{\mathcal{I}}g(p)dp$ for $p\in \mathcal{I}$. Therefore, the results extend directly to this situation. 
Moreover, the $p$-curve constructed from a finite aggregation of different tests satisfying the assumptions of Theorem \ref{thm:monotonicity} is continuously differentiable and non-increasing. 

The assumptions of Theorem \ref{thm:monotonicity} directly suggest $p$-curves for which the results of Theorem \ref{thm:monotonicity} fail. For example, when the tests are non-similar, the $p$-curve can be non-monotonic in the absence of $p$-hacking, which arises through a violation of Assumption \ref{ass:sufficient_condition}. 
To illustrate, consider testing $H_0:h\le 0$ against $H_1:h> 0$ using a (non-similar) one-sided $t$-test, where $f$ is the density of the $\mathcal{N}(0,1)$ distribution and $f_{h}$ is the density of the $\mathcal{N}(h,1)$ distribution. It follows that $f'(x)/f(x)=-x$ and $f'_{h}(x)/f_{h}(x)=-(x-h)$, such that Assumption \ref{ass:sufficient_condition} holds when $h \ge 0$ but is violated when $h<0$.
Thus, when the weight in $\Pi$ on $h<0$ is large enough, the $p$-curve can be non-monotonic or increasing. For example, suppose that $\Pi$ is a normal distribution with mean $\mu$ and variance $1$, which places some mass on $h<0$, mixing increasing and decreasing $p$-curves. Figure \ref{fig:p_curves_non_similar} shows that the resulting $p$-curve is non-increasing when $\mu=0$ and non-monotonic when $\mu=-2.5$. 
\begin{figure}[ht]
     \centering
         \includegraphics[width=0.45\textwidth,trim=0 .5cm 0 0]{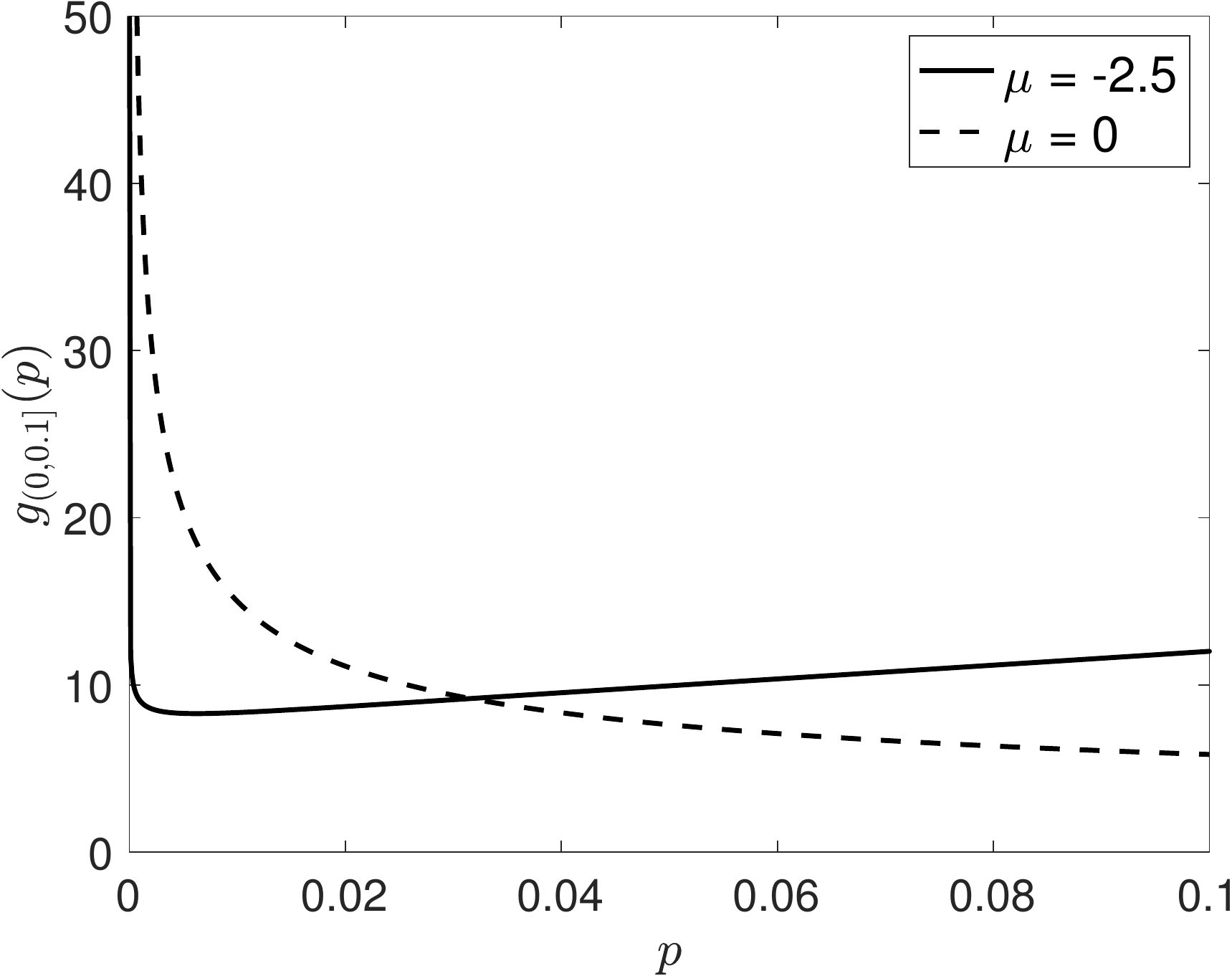}
              \caption{$P$-curves based on non-similar one-sided $t$-tests on $(0,0.1]$. The distribution of true effects $\Pi$ is a normal distribution with mean $\mu$ and variance $1$.}
         \label{fig:p_curves_non_similar}
\end{figure}

\section{The p-curve based on t-tests}
\label{sec:ttestspcurve}

We now show that for the leading case where $p$-curves are generated from $t$-tests with exact or asymptotic normal distributions, there are additional previously unknown testable restrictions. These restrictions allow us to develop more powerful statistical tests for $p$-hacking (see Section \ref{sec: tests t-tests}). In particular, these tests have power in situations where $p$-hacking does not lead to a violation of non-increasingness.

Consider first the problem of testing a one-sided hypothesis 
\begin{eqnarray}
H_0: h=0 \qquad \text{against} \qquad  H_1: h>0, \label{eq:one_sided_t}
\end{eqnarray}
where $h$ is a scalar, $\mathcal{H}_0=\{0\}$, and $\mathcal{H}_1=(0,\infty)$. We assume that $T\sim \mathcal{N}(h,1)$. This holds when using one-sided $t$-tests to test a hypothesis concerning a scalar parameter $\theta$: $H_0: \theta=\theta_0$ against $H_1: \theta > \theta_0.$ Let $\sqrt{N}\left(\hat\theta - \theta \right) \sim \mathcal{N}\left(0,\sigma^2\right)$, where $\hat\theta$ is an estimator of $\theta$ based on $N$ observations and $\sigma^2$ is assumed to be known. Denote the usual $t$-statistic as $\hat{t}$ and set $T=\hat{t}$. Defining $h:=\sqrt{N}\left((\theta-\theta_0)/\sigma\right)$ this fits \eqref{eq:one_sided_t}. More generally, testing problems with limiting normal experiments employed to test hypotheses of the form \eqref{eq:one_sided_t} are common in empirical work (e.g., a one-sided test of a regression parameter using normal critical values).  

The chosen null distribution is the standard normal distribution, $F=\Phi$. A level $p$ test rejects the null hypothesis when $T$ is larger than $cv_1(p):=\Phi^{-1}\left(1-p \right)$. Note that $cv_1(p)\ge 0$ for $p\in\left(0,1/2\right]$. Then $\beta\left(p,h \right) = 1-\Phi \left(cv_1(p)-h  \right)$ and the CDF of $p$-values is  
\begin{equation}
G_1(p)=1-\int_{{[0,\infty)}} \Phi\left(cv_1(p)-h  \right)d\Pi(h). \label{eq: CDF one-sided t-test}
\end{equation}

We also consider the two-sided version of this test. Here the hypothesis is
\begin{eqnarray}
H_0: h=0 \qquad \text{against} \qquad  H_1: h\ne 0 \label{eq:two_sided_t}
\end{eqnarray}
with $\mathcal{H}_0=\{0\}$ and $\mathcal{H}_1= \mathbb{R}\backslash \{0\}$. The two-sided test statistic $T$ is assumed to have a folded normal distribution. This holds when using a two-sided $t$-test with $T=|\hat{t}|$ for testing a two-sided hypothesis about $\theta$ : $H_0: \theta=\theta_0$ against $H_1: \theta \ne  \theta_0$.
More generally, testing problems with limiting normal experiments employed to test hypotheses of the form \eqref{eq:two_sided_t} are also common in empirical work.     

The chosen null distribution is the half normal distribution with scale parameter $1$. A level $p$ test rejects the null hypothesis when $T$ is larger than  $cv_2(p):=\Phi^{-1}\left(1-\frac{p}{2}\right)$. The CDF of the $p$-values is 
\begin{eqnarray}
G_2(p)
=2-\int_{\mathbb{R}} \left[ \Phi \left(cv_2(p)-h \right)+\Phi \left(cv_2(p)+h \right) \right] d\Pi(h). \label{eq: CDF two-sided t-test}
\end{eqnarray}

In addition to the results of Section \ref{sec:p_curve_absence}, previously unknown testable restrictions for $p$-curves based on $t$-tests follow from the shape of the power functions for these tests. These additional restrictions enable us to better pin down the space of potential $p$-curves when there is no $p$-hacking, allowing us to construct more powerful statistical tests for $p$-hacking. They also enable distinguishing non-increasing $p$-curves, which can arise from certain types of $p$-hacking, from curves where there is no $p$-hacking.

The $p$-curve based on one-sided $t$-tests testing hypothesis \eqref{eq: CDF one-sided t-test} is 
\begin{equation}
g_1(p) = \int_{{[0,\infty)}}\exp\left(hcv_1(p)-\frac{h^2}{2}\right)d\Pi(h).\label{eq: p-curve one-sided t-test}
\end{equation}
For two-sided $t$-tests testing hypothesis \eqref{eq: CDF two-sided t-test}, the $p$-curve is
\begin{equation}
g_2(p) = \int_{\mathbb{R}}\frac{1}{2}\left[\exp\left(hcv_2(p)-\frac{h^2}{2}\right)+\exp\left(-hcv_2(p)-\frac{h^2}{2}\right)\right]d\Pi(h).\label{eq: p-curve two-sided t-test}
\end{equation}
Our next theorem shows that the $p$-curves \eqref{eq: p-curve one-sided t-test} and \eqref{eq: p-curve two-sided t-test} are completely monotone. A function $\xi$ is completely monotone on an interval $\mathcal{I}$ if $0\le (-1)^{k}\xi^{(k)}(x)$ for every $x\in \mathcal{I}$ and all $k=0,1,2,\dots$, where $\xi^{(k)}$ is the k$^{th}$ derivative of $\xi$. 

\begin{theorem}[Complete monotonicity]\label{thm: complete monotonicity} (i) The $p$-curve $g_1$ is completely monotone on $\left(0,1/2\right]$. (ii)
The $p$-curve $g_2$ is completely monotone on $(0,1)$.
\end{theorem}
Complete monotonicity yields additional restrictions that
can be exploited to improve the power of statistical tests for $p$-hacking. Whilst available for one- and two-sided $t$-tests, not all tests yield completely monotonic $p$-curves. For example, a direct calculation shows that complete monotonicity may fail for tests based on $\chi^2$ distributions with more than two degrees of freedom (e.g., Wald tests).

The next theorem presents additional testable restrictions in the form of upper bounds on the $p$-curves and their derivatives.

\begin{theorem}[Upper bounds]\label{thm: bounds}
\text{ }
\begin{enumerate}\setlength\itemsep{-1.5pt}
\item[(i)] The $p$-curves $g_1$ and $g_2$ are bounded from above:
\begin{eqnarray}
g_1(p)&\le& 1_{\{p\le1/2\}}\exp\left(\frac{cv_1(p)^2}{2}\right)+1_{\{p>1/2\}}=:\mathcal{B}^{(0)}_1(p),\label{eq:one_sided_bound}\\
  g_2(p)&\le&1_{\{p<2(1-\Phi(1))\}}\tilde{\mathcal{B}}^{(0)}_2+1_{\{p\ge 2(1-\Phi(1))\}}=:\mathcal{B}^{(0)}_2(p),\label{eq:two_sided_bound}
\end{eqnarray}
where
\begin{eqnarray*}
\tilde{\mathcal{B}}^{(0)}_2(p)&:=&\frac{1}{2}\left[\exp\left(h^*(p)cv_2(p)-\frac{h^{*}(p)^2}{2}\right)+\exp\left(-h^*(p)cv_2(p)-\frac{h^{*}(p)^2}{2}\right)\right]\\
&\le&\exp\left(\frac{cv_2(p)^2}{2}\right),
\end{eqnarray*}
and $h^*(p)$ is the non-zero solution to $$\varphi(cv_2(p), h):=(cv_2(p)-h)\exp(cv_2(p)h)-(cv_2(
p)+h)\exp(-cv_2(p)h)=0.$$
\item[(ii)] The derivatives of $g_1$ and $g_2$ are bounded from above. For $s=1,2$ and $k=1,2,3,\dots$, then $(-1)^kg^{(k)}_{s}(p)\le \mathcal{B}^{(k)}_{s}(p),$
where $\mathcal{B}^{(k)}_{s}$ is defined in Appendix \ref{app: proof bounds}.
\end{enumerate}
\end{theorem}
As with the results in Theorem \ref{thm: complete monotonicity}, the results in Theorem \ref{thm: bounds} yield additional restrictions, allowing more powerful tests for $p$-hacking.\footnote{One can use similar arguments as in Theorem \ref{thm: bounds} to derive bounds for $p$-curves based on other specific tests such as Wald tests.} The bounds in Theorem \ref{thm: bounds} do not only rule out large humps around significance cutoffs such as 0.01, 0.05, and 0.1 but also restrict the magnitude of the $p$-curves near zero. For the two-sided test, tests for $p$-hacking can be either constructed using the sharper (but not explicit) bound $\tilde{\mathcal{B}}^{(0)}_{2}(p)$ or the simpler explicit bound $\exp\left(\frac{cv_2(p)^2}{2}\right)$. 

The bounds of Theorem \ref{thm: bounds} are particularly useful when $p$-hacking fails to induce an increasing $p$-curve, a situation where tests based on non-increasingness of the $p$-curve have no power. Intuitively we might suspect this happens when all researchers $p$-hack but this simply shifts mass of the $p$-curve to the left, rather than inducing humps. A concrete example is when researchers run a finite number of $M>1$ independent analyses and report the smallest $p$-value, for example, when engaging in specification search across independent subsamples or data sets. The resulting $p$-curve under $p$-hacking is $g^{p}(p;M) = M(1-G^{np}(p))^{M-1}g^{np}(p)$, where $G^{np}$ and $g^{np}$ are the CDF and density of $p$-values in the absence of $p$-hacking.\footnote{This generalizes the example in \citet{ulrich2015p}, who studied the special case where all null hypotheses are true such that $G(p)=p$.} Note that $g^{p}$ is non-increasing (completely monotone) whenever $g^{np}$ is non-increasing (completely monotone).\footnote{Since the products of completely monotone functions are completely monotone, complete monotonicity of $g^{p}(p;M)$ follows from complete monotonicity of $1-G^{np}(p)$ and $g^{np}(p)$.} Thus, $g^p$ will not violate the testable implications of Theorems \ref{thm:monotonicity}--\ref{thm: complete monotonicity}, so tests based on these restrictions do not have power. However, $g^{p}$ can violate the bounds in Theorem \ref{thm: bounds} whenever $M(1-G^{np}(p))^{M-1}>1$. For example, consider the one-sided case and let $\Pi$ be a half-normal distribution with scale parameter 1. Figure \ref{fig: violation bound} shows that $g^{p}$ violates the upper bound in Theorem \ref{thm: bounds} to an extent that depends on $M$. 

\begin{figure}[H]
 \centering
     \includegraphics[width=0.45\textwidth,trim=0 0.5cm 0 0]{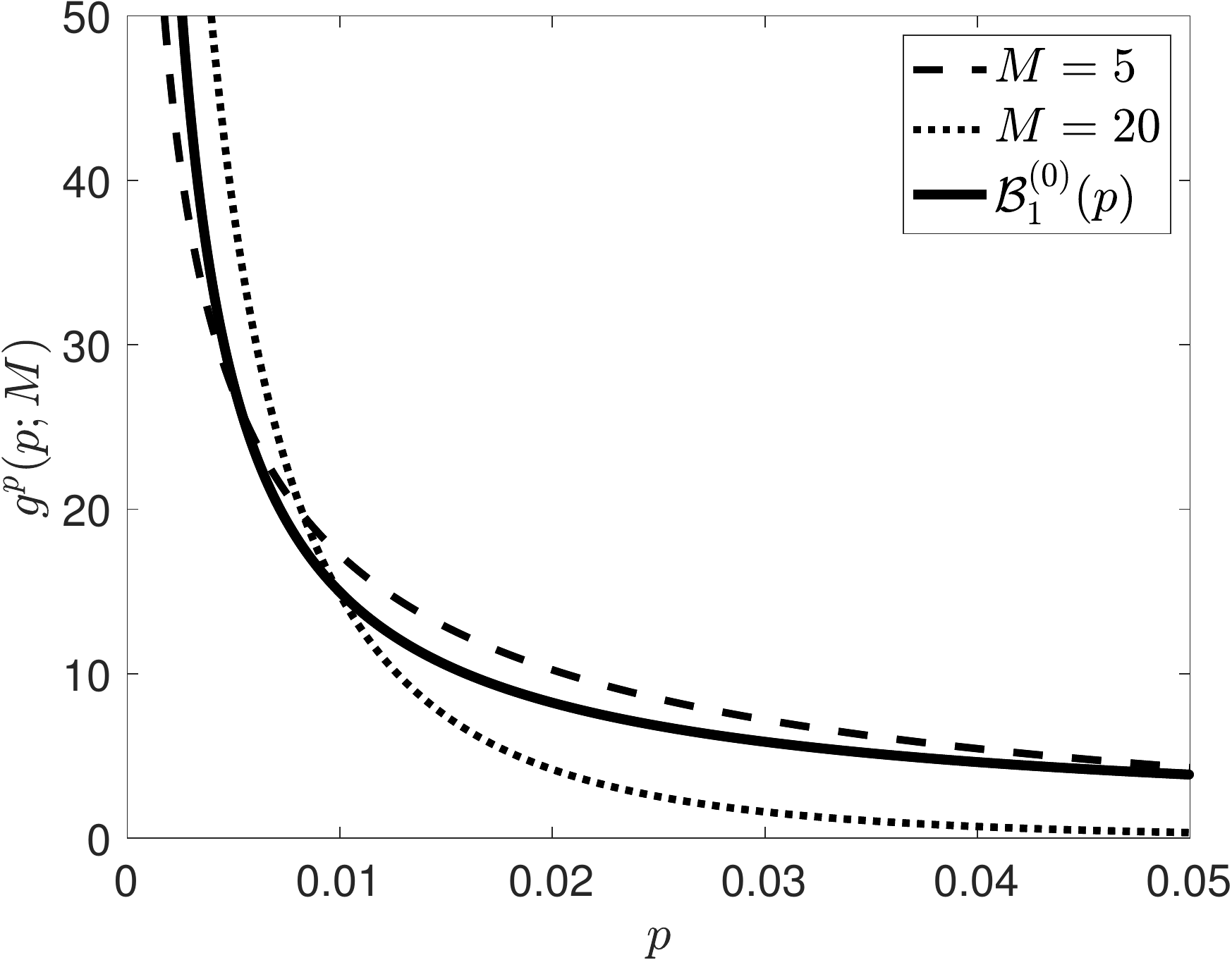}
     \caption{Comparison of the $p$-curve from specification search based on one-sided $t$-tests and the upper bound in Equation \eqref{eq:one_sided_bound}.}

     \label{fig: violation bound}
\end{figure}

Upper bounds also help with testing for $p$-hacking with non-similar tests. In Section \ref{sec:p_curve_absence}, we show that non-increasingness may fail for non-similar one-sided $t$-tests, in which case tests of $p$-hacking based on non-increasingness may well reject because of non-similarity rather than $p$-hacking. Since upper bounds can also be derived for non-similar tests, we can still use bounds on the $p$-curve and its derivatives to test for $p$-hacking.\footnote{For instance, for $p\le 1/2$, the upper bound on the $p$-curve for non-similar one-sided $t$-tests coincides with that in Part (i) of Theorem \ref{thm: bounds}.}

Finally, the characterizations in Theorems \ref{thm: complete monotonicity}--\ref{thm: bounds} imply related characterizations of $p$-curves over subintervals $\mathcal{I}\subset (0,1)$, $g_{s,\mathcal{I}}(p)=g_s(p)/\int_{\mathcal{I}}g_s(p)dp$. In particular, complete monotonicity of $g_s$ implies the complete monotonicity of $g_{s,\mathcal{I}}$, because the sign of $g^{(k)}_{s,\mathcal{I}}$ equals the sign of $g_s^{(k)}$ for $k=0,1,2\dots$. Moreover, (conservative) upper bounds on $g_{s,\mathcal{I}}(p)$ for $\mathcal{I}=(0,\alpha]$ are given by the upper bounds in Theorem \ref{thm: bounds}, re-scaled by $\alpha$ since $G_s(\alpha)\ge \alpha$ for $s=1,2$. 

\section{Statistical tests for p-hacking}
\label{sec:tests}
Here we consider tests for $p$-hacking based on a sample of $n$ $p$-values. We consider three types of tests that differ with respect to the specification of the null hypothesis (the null space of $p$-curves). As a result, the different tests will differ with respect to the violations of the null of no $p$-hacking that they are able to detect.

In the absence of publication bias, our tests are tests for $p$-hacking; when there is also publication bias, they are joint tests for $p$-hacking and publication bias in general.

\subsection{Tests for non-increasingness of the p-curve}
\label{sec: tests non-increasingness}
Theorem \ref{thm:monotonicity} shows that, under general conditions, the $p$-curve is non-increasing. Consider the following testing problem
\begin{equation}
H_0: g ~\text{is non-increasing} \qquad \text{against} \qquad H_1: g ~\text{is not non-increasing}. \label{eq:testing_monotonicity}
\end{equation}
Popular tests based on hypothesis testing problem \eqref{eq:testing_monotonicity} include the Binomial test \citep[e.g.,][]{simonsohn2014p,head2015extent} and Fisher's test \citep{simonsohn2014p}. Here we describe two alternative and more powerful tests.

\medskip

\noindent \textbf{Histogram-based tests.}
Let $0=x_0<x_1<\cdots< x_J=1$ be an equidistant partition of the unit interval. Define the population proportions as $\pi_j := \int_{x_{j-1}}^{x_j}g(p)dp, ~ j=1,\dots, J$. When $g$ is non-increasing, $\Delta_j :=\pi_{j+1} - \pi_j$ is non-positive for all $j=1,\dots,J-1$. Thus, the null hypothesis in testing problem \eqref{eq:testing_monotonicity} can be reformulated as $H_0: \Delta_j\le 0$ for all $j=1,\dots,J-1$.
To test this hypothesis, we apply the conditional chi-squared test of \citet{cox2020simple}. We describe the implementation of this test in Section \ref{sec: tests t-tests} and Appendix \ref{app: additional details}, where we propose more general tests that nest the histogram-based test for non-increasingness.

\medskip

\noindent \textbf{LCM test based on concavity of the CDF of $p$-values.}
Under the null hypothesis \eqref{eq:testing_monotonicity}, the CDF of $p$-values is concave. This observation allows us to apply tests based on the least concave majorant (LCM) \citep[e.g.,][]{carolan2005,beare2015nonparametric,fang2019refinements}. LCM-based tests assess concavity of the CDF based on the distance between the empirical CDF of $p$-values, $\hat{G}$, and its LCM, $\mathcal{M}\hat{G}$, where $\mathcal{M}$ is the LCM operator.\footnote{For a function $f$, the LCM operator is defined as $\mathcal{M}f=\inf \{g: g \text{ is concave and } f\le g\}$ \citep[e.g.,][Definition 2.1]{beare2015nonparametric}.} We consider the test statistic $T = \sqrt{n}\| \mathcal{M}\hat{G} - \hat{G} \|_\infty$. The uniform distribution is least favorable for LCM tests \citep[e.g.,][]{kulikov2008distribution,beare2021least}, in which case $T$ converges weakly to $\| \mathcal{M}B - B \|_\infty$, where $B$ is a standard Brownian Bridge on $[0,1]$. 

\subsection{Tests for continuity}
\label{sec: tests continuity}
Theorem \ref{thm:monotonicity} shows that the $p$-curve is continuous in the absence of $p$-hacking. Tests for continuity of the $p$-curve at significance thresholds $\alpha$ such as $\alpha=0.05$, thus, provide an alternative to the tests based on non-increasingness of the $p$-curve. Consider the following testing problem:
\begin{equation}
H_0:~\lim_{p \uparrow \alpha}g(p)=\lim_{p \downarrow \alpha}g(p)\qquad \text{against} \qquad H_1:~\lim_{p \uparrow \alpha}g(p)\ne \lim_{p \downarrow \alpha}g(p)\label{eq:testing_continuity}
\end{equation}
Testing \eqref{eq:testing_continuity} requires estimating two densities at the boundary point $\alpha$. Traditional kernel density estimators are not suitable for this task because they suffer from boundary bias \citep[e.g.,][]{karunamuni2005boundary}. A popular approach to overcome this problem is to use local linear density estimators that rely on prebinning the data \citep[e.g.,][]{mccrary2008manipulation}. We apply the density discontinuity test of \citet{cattaneo2020simple} with data-driven bandwidth selection \citep{cattaneo2020rddensity}, which is based on boundary adaptive local polynomial density estimators and avoids prebinning. 

\subsection{Tests for K-monotonicity and upper bounds} \label{sec: tests t-tests}

Theorem \ref{thm: complete monotonicity} shows that $p$-curves based on $t$-tests are completely monotone, and Theorem \ref{thm: bounds} establishes upper bounds on the $p$-curves and their derivatives. Here we develop tests based on these testable restrictions. 

We say a function $\xi$ is $K$-monotone on some interval $\mathcal{I}$ if $0\le (-1)^{k}\xi^{(k)}(x)$ for every $x\in \mathcal{I}$ and all $k=0,1,\dots,K$, where $\xi^{(k)}$ is the k$^{th}$ derivative of $\xi$. By definition, a completely monotone function is $K$-monotone. Consider the null hypothesis
\begin{align}
H_0:g_s\text{ is $K$-monotone and } (-1)^kg_s^{(k)}\le \mathcal{B}^{(k)}_s,\text{ for } k=0,1,\dots,K, \label{eq:testing_additional_constraints}
\end{align}
where $s=1$ for one-sided $t$-tests, $s=2$ for two-sided $t$-tests, and $\mathcal{B}^{(k)}_s$ is defined in Theorem \ref{thm: bounds}. 
Hypothesis \eqref{eq:testing_additional_constraints} implies restrictions on the population proportions $\boldsymbol{\pi}:=(\pi_1,\dots,\pi_J)'$, which can be expressed as $H_0: A\boldsymbol{\pi}_{-J}\le b$, where $\boldsymbol{\pi}_{-J} := (\pi_1,\dots, \pi_{J-1})'$.\footnote{The upper bounds on $\boldsymbol{\pi}$ implied by hypothesis \eqref{eq:testing_additional_constraints} are not sharp in general. Sharp bounds can be obtained by directly extremizing the proportions and their differences; see Appendix \ref{app: details cox and shi test 1}.} The matrix $A$ and vector $b$ are defined in Appendix \ref{app: details cox and shi test 2}.\footnote{We use $\boldsymbol{\pi}_{-J}$ because the variance matrix of the estimator of $\boldsymbol{\pi}$ is singular by construction and we want to express the left-hand side of our moment inequalities as a combination of ``core'' moments.}

We estimate $\boldsymbol{\pi}_{-J}$ using the sample proportions $\hat{\boldsymbol{\pi}}_{-J}$.\footnote{Given a sample of $n$ $p$-values, $\{P_i\}_{i=1}^n$, the sample proportions are defined as $\hat\pi_i=\frac{1}{n}\sum_{i=1}^n1\{x_{i-1}< P_i \le x_i\}$, $i=1,\dots,J$.}  This estimator is $\sqrt{n}$-consistent and asymptotically normal with mean $\boldsymbol{\pi}_{-J}$ and non-singular (if all proportions are positive) covariance matrix $\Omega = \text{diag}\{{\pi}_1, \dots, \pi_{J-1}\} - \boldsymbol{\pi}_{-J}\boldsymbol{\pi}_{-J}'$.
Following \citet{cox2020simple}, we test the null by comparing $T = \inf_{q:\: Aq\le b}n (\hat{\boldsymbol{\pi}}_{-J} - q)'\hat{\Omega}^{-1} (\hat{\boldsymbol{\pi}}_{-J} - q) \label{eq: CS_stat}$ to the critical value from a $\chi^2$ distribution with $rank(\hat{A})$ degrees of freedom, where $\hat{A}$ is the matrix formed by the rows of $A$ corresponding to active inequalities.

\section{Empirical applications}
\label{sec:application}

The analyses were done using R \citep{R20} and Stata \citep{stata2019}.

\subsection{P-hacking in economics journals}
\label{sec:brodeur_application}

Here we reanalyze the data collected by \citet{brodeur2016}, which contain information about 50,078 $t$-tests from 641 papers published in the AER, QJE, and JPE 2005--2011 \citep{brodeur2016data}. We convert $t$-statistics into $p$-values associated with two-sided $t$-tests based on the standard normal distribution.\footnote{The original data contain $p$-values for less than 10\% of observations. Where available, we work with the reported $p$-values.} After excluding observations with missing information, there are 49,838 tests from 640 papers.

Because the $p$-values may be correlated within papers, we use cluster-robust estimators of the variance of the sample proportions for the \citet{cox2020simple} tests. In addition, we apply all tests to random subsamples with one $p$-value per paper, allowing us to use exact tests in the presence of within-paper correlation. 
To test for $p$-hacking, we focus on $p$-values smaller than $0.15$. We consider a Binomial test on $[0.04,0.05]$, Fisher's test, a histogram-based test for non-increasingness (CS1), a histogram-based test for $2$-monotonicity and bounds on the $p$-curve and the first two derivatives (CS2B), the LCM test, and a density discontinuity test at 0.05.\footnote{For the Binomial test, we split $[0.04,0.05]$ into two subintervals $[0.04,0.045]$ and $(0.045,0.05]$. Under the null of no $p$-hacking, the fraction of $p$-values in $(0.045,0.05]$ should be smaller than or equal to 0.5, which we assess using an exact Binomial test. For CS1 and CS2B, we use 30 bins when testing based on all $p$-values and 15 bins when testing based on random subsamples of $p$-values.}

Figure \ref{fig:brodeur_015} shows the results before and after de-rounding and based on the full sample and random subsamples. There is a large number of very small $p$-values, which is sometimes interpreted as indicative of evidential value (e.g., \citet{simonsohn2014p}; in our notation, this is a large mass of $\Pi$ away from zero). The data exhibit a noticeable mass point at $\hat{t}=2$ (there are 427 such observations), which translates into a mass point in the $p$-curve at $p=0.046$.\footnote{This mass point could be due to low precision reporting \citep{brodeur2016}, but also due to $p$-hacking, publication bias, or a combination thereof.}
To analyze the impact of rounding, we also apply the tests to the de-rounded data provided by \citet{brodeur2016}.\footnote{The de-rounded data were constructed by randomly redrawing estimates and standard errors; see Section II in \citet{brodeur2016} for a detailed description.} 

\begin{figure}[ht]
\scriptsize
     \centering
     \includegraphics[width=0.4\textwidth,trim=0 0cm 0 0]{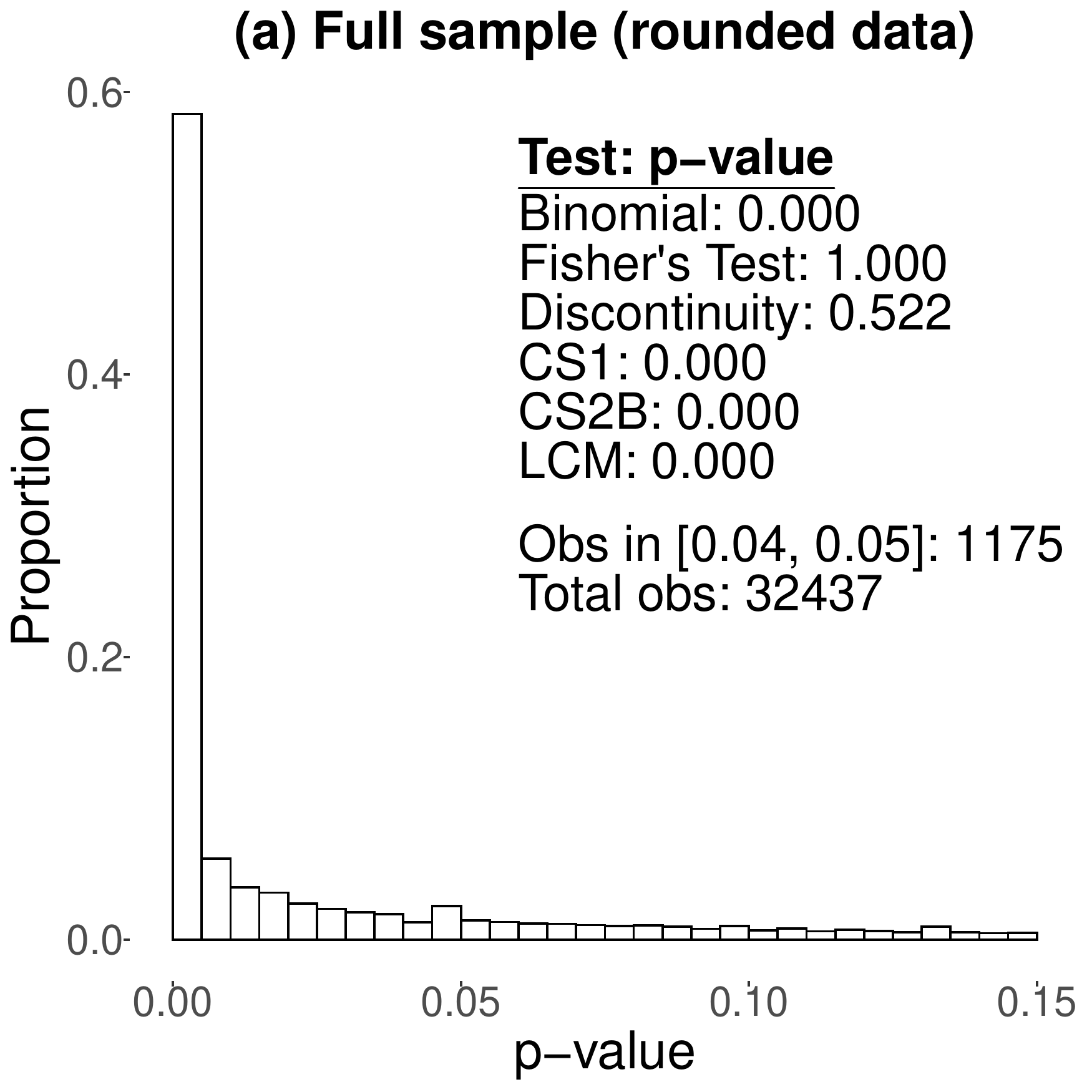}
     \includegraphics[width=0.4\textwidth,trim=0 0cm 0 0]{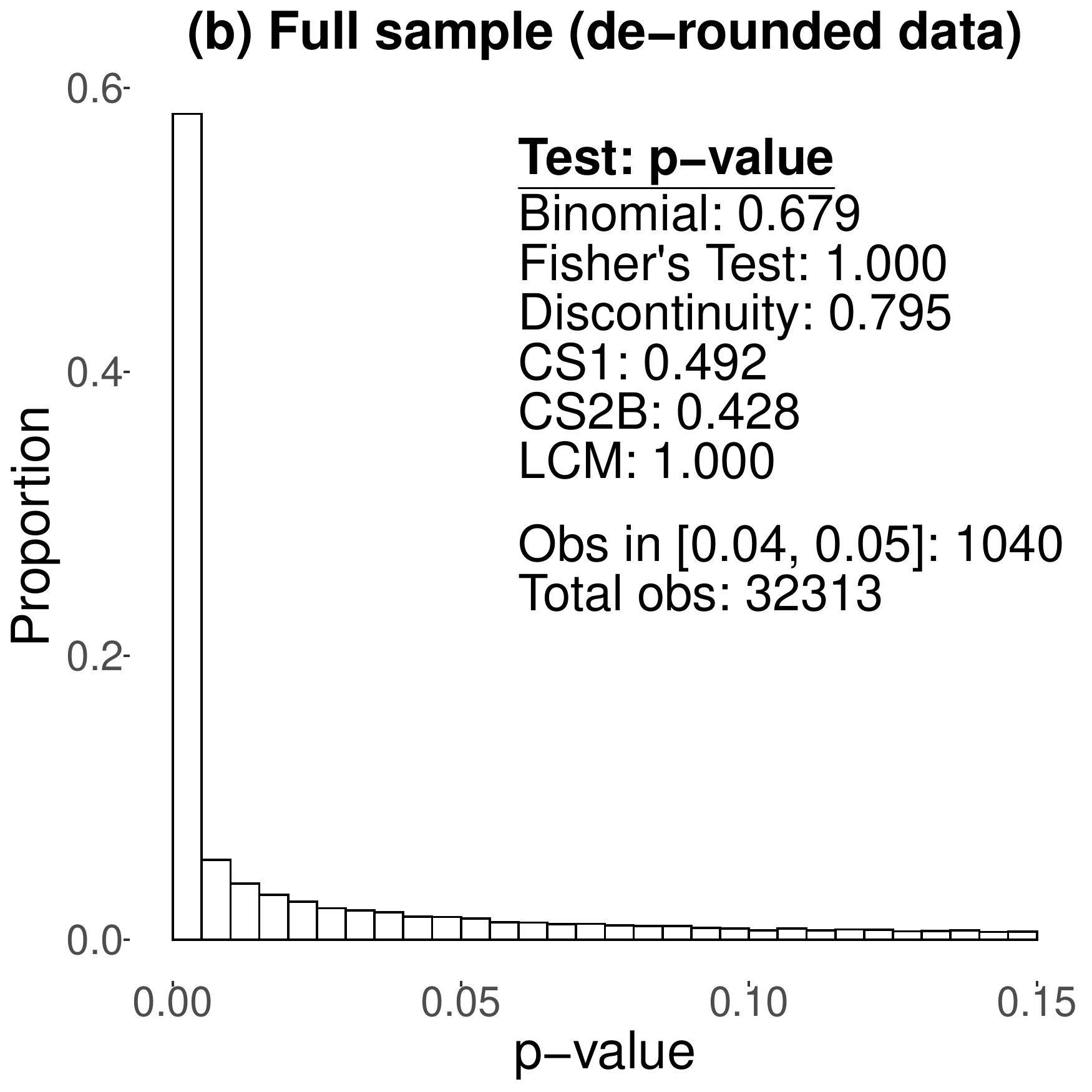}
     \:
    \includegraphics[width=0.4\textwidth,trim=0 0.5cm 0 0]{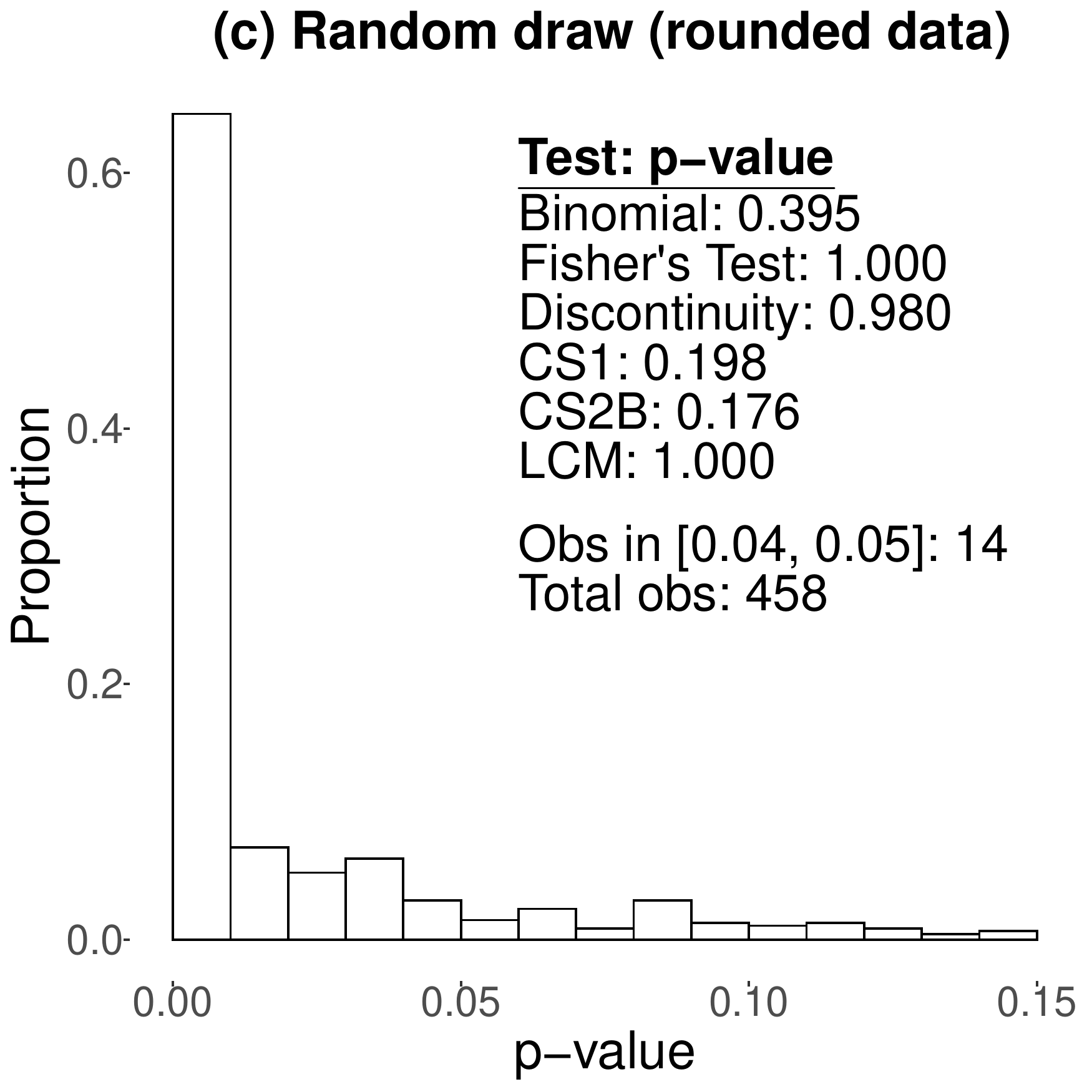}
    \includegraphics[width=0.4\textwidth,trim=0 0.5cm 0 0]{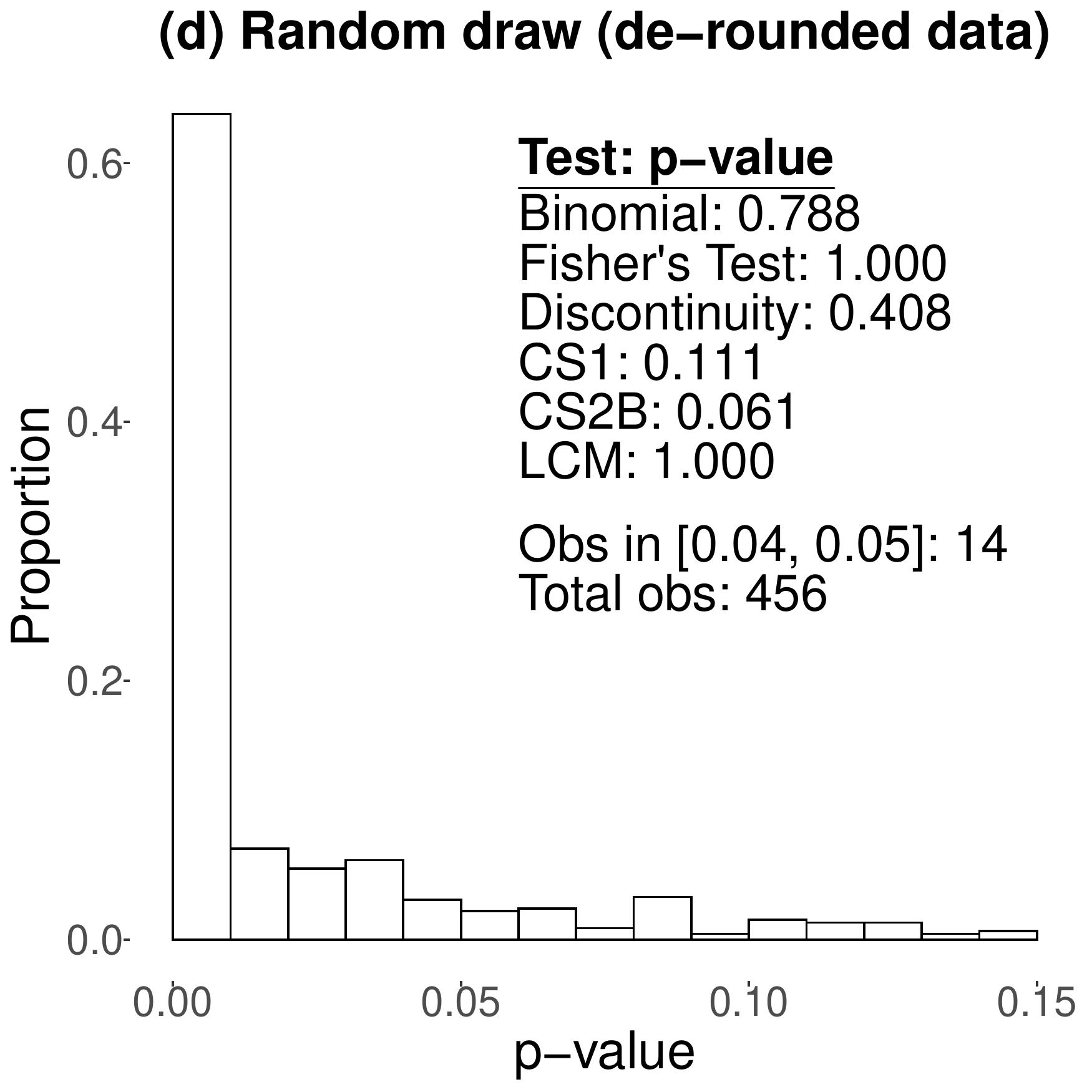}

    \caption{$P$-curves and $p$-values from testing for $p$-hacking. The tests for $p$-hacking are described in Section \ref{sec:tests}. Data: \citet{brodeur2016data}.}
    
\normalsize        
\label{fig:brodeur_015}
\end{figure}

In what follows, we say that a test rejects the null of no $p$-hacking if its $p$-value is smaller than $0.1$. Based on the original raw (rounded) data on all $p$-values, all tests reject the null except Fisher's test and the density discontinuity test. There are no rejections based on the random subsample, suggesting that the tests may be underpowered in small samples.

We find different results based on the de-rounded data.\footnote{Note that the (sub)sample sizes for the rounded and de-rounded data differ due to de-rounding.} There are no rejections based on the full sample of $p$-values. This finding suggests that the rejections based on the raw data are mainly due to the mass point just below $0.05$ and shows that de-rounding may substantially affect empirical conclusions.

Based on the random subsample of de-rounded $p$-values, only the CS2B test rejects the null of no $p$-hacking. The CS1 test comes close to rejecting ($p=0.11$). These two tests yield the smallest $p$-values across all four samples.

\subsection{P-hacking across different disciplines}
\label{sec: head application}

Here we reanalyze the data collected by \citet{head2015extent}, which contain $p$-values obtained from text-mining open access papers in the PubMed database \citep{head2016data}. There are $p$-values from 21 different disciplines. We focus on biology, chemistry, education, engineering, medical and health sciences, and psychology and cognitive science. The data contain $p$-values from the abstracts and the results sections in the main text. We use $p$-values from the results sections, allowing us to work with larger samples and present results for $p$-values smaller than $0.15$.

Since the data do not only contain $t$-tests, we consider tests based on non-increasingness and continuity of the $p$-curve (Theorem \ref{thm:monotonicity}):
a Binomial test on $[0.04,0.05]$, Fisher's test, a histogram-based test for non-increasingness (CS1), the LCM test, and a density discontinuity test at 0.05.\footnote{For CS1, we use 60 bins (all data) and 30 bins (random subsamples) for biological and medical and health sciences given the large sample sizes, and 30 and 15 bins for the other disciplines.} To account for within-paper dependence of $p$-values, we use a cluster-robust variance estimator for the CS1 test, and also present results based on random subsamples with one $p$-value per paper.

\begin{figure}[ht]
\begin{center}
\includegraphics[width=0.4\textwidth,trim=0 1.5cm 0 0]{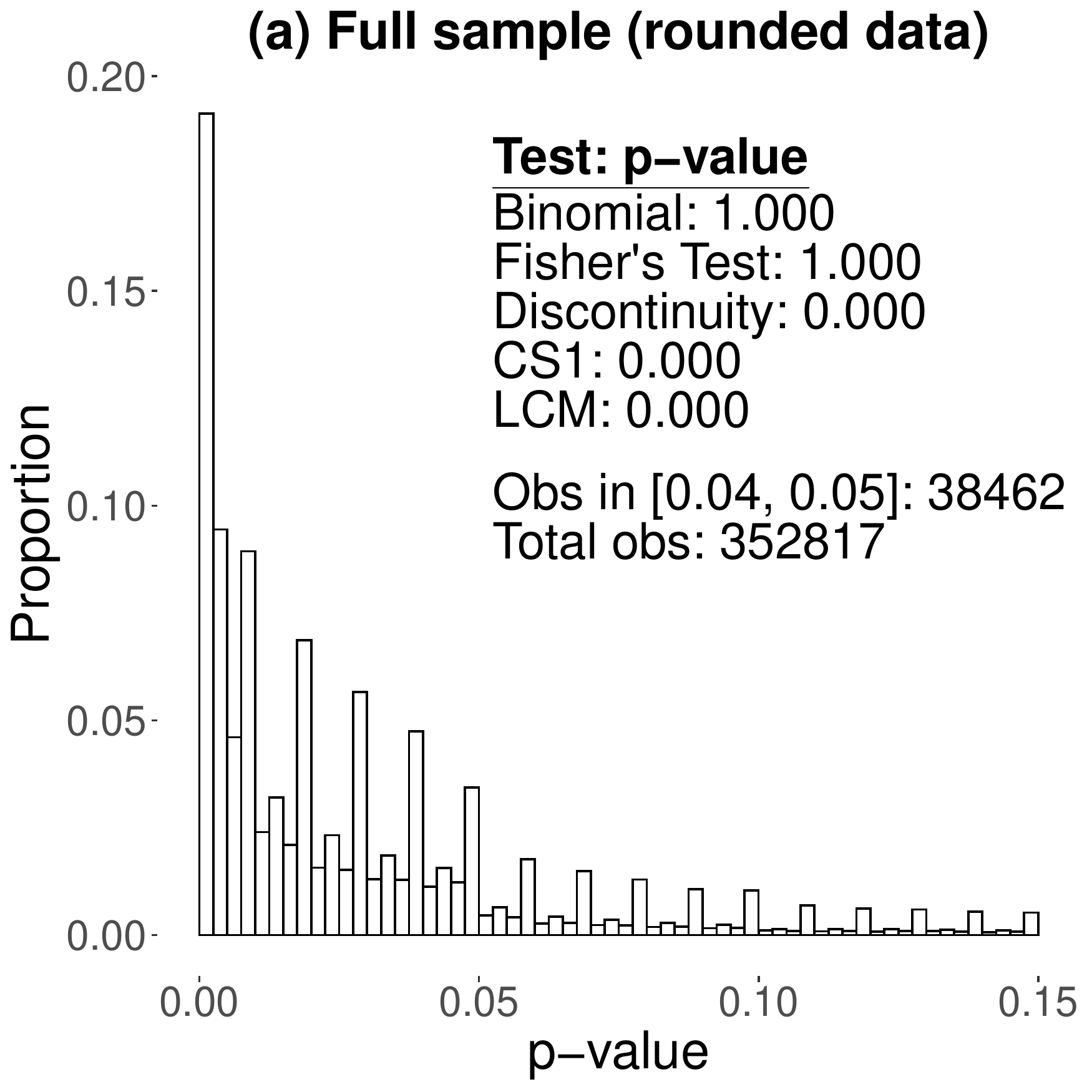}
\includegraphics[width=0.4\textwidth,trim=0 1.5cm 0 0]{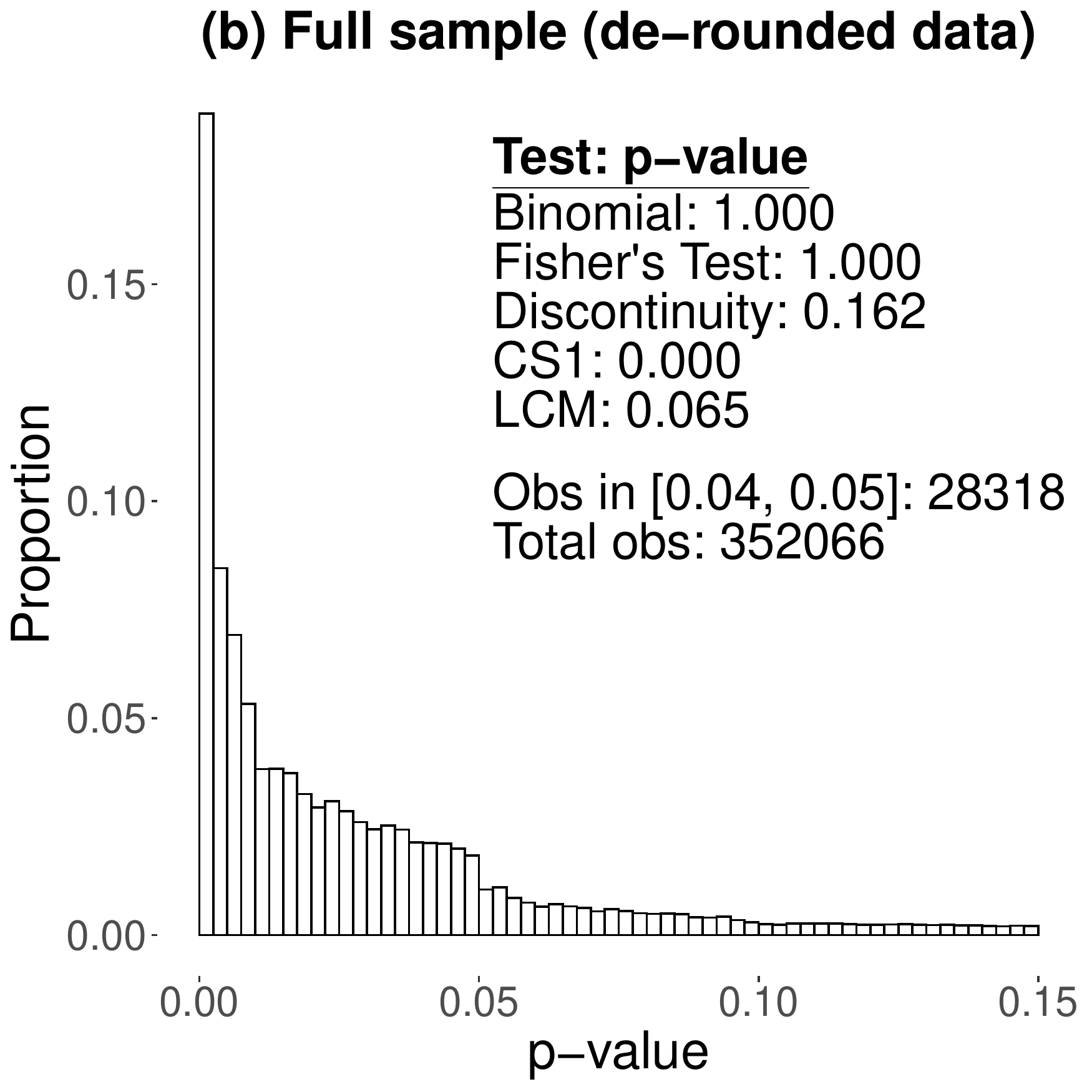}
  \end{center}
  \caption{$P$-curves and $p$-values from testing for $p$-hacking for medical and health sciences. The tests for $p$-hacking are described in Section \ref{sec:tests}. Data: \citet{head2016data}.}
\label{fig:head_med}
\end{figure}

The left panel of Figure \ref{fig:head_med} shows a histogram of the raw data on all $p$-values for the medical and health sciences (the largest subsample). A substantial fraction of $p$-values is rounded to two decimal places, which results in sizable mass points at $0.01, 0.02, \dots, 0.15$. Rounding makes the $p$-curve non-monotonic and discontinuous even in the absence of $p$-hacking and, thus, invalidates the testable restrictions in Theorem \ref{thm:monotonicity}. Therefore, we also show results based on de-rounded data.\footnote{We de-round the data as follows. To each observed $p$-value rounded up to the k$^{th}$ decimal point we add a random number generated from the uniform distribution supported on the interval $[\underline{u}, 0.5] \cdot 10^{-k}$, where $\underline{u}=0$ for zero $p$-values and $\underline{u}=-0.5$ for non-zero $p$-values.} In an earlier version of this paper \citep{elliott2020detecting}, we show that de-rounding restores the non-increasingness but not the continuity of the $p$-curve. The right panel of Figure \ref{fig:head_med} shows the impact of de-rounding on the shape of the $p$-curve. We note that density discontinuity tests are poorly suited here because rounding induces substantial discontinuities, which remain even after de-rounding. This means that rejections of the null can be either due to rounding or due to $p$-hacking. 

In what follows, define a rejection of the null of no $p$-hacking for $p$-values smaller than $0.1$. Table \ref{tab:all_15} presents the results for the full sample of $p$-values. For the original (rounded) data, the CS1 and the LCM test reject the null for all disciplines. De-rounding leads to fewer rejections. The CS1 test only rejects for biological sciences, engineering, and medical and health sciences; the LCM test rejects for medical and health sciences. This shows that rounding and de-rounding can substantially affect empirical results. The Binomial and Fisher's test do not reject the null for any discipline, which demonstrates the importance of using our more powerful tests.

\begin{table}[ht]
\caption{Testing results based on full sample of $p$-values}
\vspace{-0.5cm}

\begin{center}
\scriptsize
\begin{tabular}{@{}ccccccc@{}}\toprule \toprule
 \multirow{2}{*}{Test}&&&\multicolumn{2}{c}{Discipline}&&\\ \cmidrule{2-7}
& \multicolumn{1}{c}{\makecell{Biological \\ sciences}} & \multicolumn{1}{c}{\makecell{Chemical\\ sciences}} & \multicolumn{1}{c}{\makecell{Education}}&\multicolumn{1}{c}{\makecell{Engineering}} &  \multicolumn{1}{c}{\makecell{Medical and \\ health sciences}}  & \multicolumn{1}{c}{\makecell{Psychology and  \\ cognitive sciences}}\\ 
 \midrule

&&&\multicolumn{2}{c}{Rounded}&&\\ \cmidrule{2-7}
 \multicolumn{1}{l}{Binomial}	&	1.000	&	0.342	&	0.975	&	0.999	&	1.000	&	1.000	\\
\multicolumn{1}{l}{Fisher's Test}	&	1.000	&	1.000	&	1.000	&	1.000	&	1.000	&	1.000	\\
\multicolumn{1}{l}{Discontinuity}	&	0.000	&	0.000	&	0.159	&	0.000	&	0.000	&	0.172	\\
\multicolumn{1}{l}{CS1}&	0.000	&	0.000	&	0.000	&	0.000	&	0.000	&	0.000	\\
\multicolumn{1}{l}{LCM}	&	0.000	&	0.000	&	0.000	&	0.000	&	0.000	&	0.000	\\\midrule
\multicolumn{1}{l}{Obs in [0.04, 0.05]}	&	7692	&	296	&	220	&	396	&	38462	&	1621	\\ 
\multicolumn{1}{l}{Total obs}	&	74746	&	2631	&	1993	&	3262	&	352817	&	15189	\\
\midrule
%
&&&\multicolumn{2}{c}{De-rounded}&&\\ \cmidrule{2-7}
\multicolumn{1}{l}{Binomial}	&	0.993	&	0.133	&	0.467	&	0.975	&	1.000	&	0.811	\\
\multicolumn{1}{l}{Fisher's Test}	&	1.000	&	1.000	&	1.000	&	1.000	&	1.000	&	1.000	\\
\multicolumn{1}{l}{Discontinuity}	&	0.005	&	0.117	&	0.245	&	0.849	&	0.162	&	0.406	\\
\multicolumn{1}{l}{CS1}	&	0.028	&	0.530	&	0.884	&	0.084	&	0.000	&	0.836	\\
\multicolumn{1}{l}{LCM}	&	0.936	&	1.000	&	1.000	&	1.000	&	0.065	&	0.653	\\ \midrule
\multicolumn{1}{l}{Obs in [0.04, 0.05]}	&	5720	&	234	&	144	&	250	&	28318	&	1161	\\ 
\multicolumn{1}{l}{Total obs}	&	74550	&	2628	&	1988	&	3258	&	352066	&	15130	\\	\bottomrule \bottomrule
\end{tabular}
\end{center}
\label{tab:all_15}
 \scriptsize{\textit{Notes:} Table reports $p$-values from applying different tests for $p$-hacking based on the full sample of $p$-values for rounded and de-rounded data. The tests for $p$-hacking are described in Section \ref{sec:tests}. Data: \citet{head2016data}.}
\end{table}

Table \ref{tab:rnd_15} shows the results based on random samples with one $p$-value per paper. We find that the CS1 test (biological sciences, engineering, medical and health sciences) and the LCM test (all disciplines except chemical sciences) reject the null based on the rounded data. None of the tests based on non-increasingness rejects the null based on the de-rounded data. A comparison to the results based on all $p$-values shows that the sample sizes required for detecting $p$-hacking may be quite large.

\begin{table}[ht]
\caption{Testing results based on random subsamples of one $p$-value per paper}

\vspace{-0.5cm}
\begin{center}
\scriptsize
\begin{tabular}{@{}ccccccc@{}}\toprule \toprule
 \multirow{2}{*}{Test}&&&\multicolumn{2}{c}{Discipline}&&\\ \cmidrule{2-7}
& \multicolumn{1}{c}{\makecell{Biological \\ sciences}} & \multicolumn{1}{c}{\makecell{Chemical\\ sciences}} & \multicolumn{1}{c}{\makecell{Education}}&\multicolumn{1}{c}{\makecell{Engineering}} &  \multicolumn{1}{c}{\makecell{Medical and \\ health sciences}}  & \multicolumn{1}{c}{\makecell{Psychology and  \\ cognitive sciences}}\\ 
 \midrule

&&&\multicolumn{2}{c}{Rounded}&&\\ \cmidrule{2-7}
 \multicolumn{1}{l}{Binomial}	&	0.510	&	0.157	&	0.439	&	0.904	&	1.000	&	0.670	\\
 \multicolumn{1}{l}{Fisher's Test}	&	1.000	&	1.000	&	1.000	&	1.000	&	1.000	&	1.000	\\
  \multicolumn{1}{l}{Discontinuity}	&	0.113	&	0.083	&	0.103	&	0.000	&	0.000	&	0.157	\\
 \multicolumn{1}{l}{CS1}	&	0.000	&	0.637	&	0.232	&	0.078	&	0.000	&	0.734	\\
 \multicolumn{1}{l}{LCM}	&	0.000	&	0.265	&	0.035	&	0.002	&	0.000	&	0.000	\\ \midrule
 \multicolumn{1}{l}{Obs in [0.04, 0.05]}	&	1482	&	63	&	42	&	85	&	6270	&	185	\\ 
  \multicolumn{1}{l}{Total obs}	&	13829	&	482	&	366	&	619	&	56892	&	1730	\\ \midrule
&&&\multicolumn{2}{c}{De-rounded}&&\\ \cmidrule{2-7}
 \multicolumn{1}{l}{Binomial} 	&	0.178	&	0.116	&	0.286	&	0.712	&	0.976	&	0.465	\\
 \multicolumn{1}{l}{Fisher's Test}	&	1.000	&	1.000	&	1.000	&	1.000	&	1.000	&	1.000	\\
  \multicolumn{1}{l}{Discontinuity}	&	0.571	&	0.085	&	0.997	&	0.287	&	0.557	&	0.637	\\
 \multicolumn{1}{l}{CS1}	&	0.992	&	0.688	&	0.481	&	0.731	&	0.872	&	0.747	\\
 \multicolumn{1}{l}{LCM}	&	1.000	&	1.000	&	1.000	&	0.999	&	0.846	&	1.000	\\ \midrule
 \multicolumn{1}{l}{Obs in [0.04, 0.05]}	&	1053	&	45	&	28	&	51	&	4536	&	128	\\ 
 \multicolumn{1}{l}{Total obs}	&	13788	&	482	&	365	&	619	&	56753	&	1716	\\	\bottomrule \bottomrule
\end{tabular}
\end{center}
\scriptsize{\textit{Notes:} Table reports $p$-values from applying different tests for $p$-hacking based on random subsamples of $p$-values for rounded and de-rounded data. The tests for $p$-hacking are described in Section \ref{sec:tests}. Data: \citet{head2016data}.}
 \label{tab:rnd_15}
\end{table}

Finally, the density discontinuity test rejects for at least three disciplines based on the full sample and the random subsamples. After de-rounding, it only rejects for biological sciences (full sample) and chemical sciences (random subsample). These rejections are expected because of the prevalence of rounding-induced discontinuities.

\section{Conclusion}
\label{sec:conclusion}

We provide theoretical foundations for testing for $p$-hacking based on the distribution of $p$-values across scientific studies. We establish general results on the $p$-curve, providing conditions under which a null set of $p$-curves can be shown to be non-increasing. For $p$-values based on $t$-tests, we derive previously unknown additional restrictions on the $p$-curve when there is no $p$-hacking. These restrictions lead to the suggestion of more powerful tests that can be used to test the absence of $p$-hacking. A reanalysis of two datasets from the literature shows that the new tests based on additional restrictions are useful in testing for $p$-hacking. 

\bibliographystyle{apalike}
\bibliography{references}

\begin{thebibliography}{}

\bibitem[Andrews and Kasy, 2019]{andrews2019identification}
Andrews, I. and Kasy, M. (2019).
\newblock Identification of and correction for publication bias.
\newblock {\em American Economic Review}, 109(8):2766--94.

\bibitem[Beare, 2021]{beare2021least}
Beare, B.~K. (2021).
\newblock Least favorability of the uniform distribution for tests of the
  concavity of a distribution function.
\newblock {\em Stat}, page e376.
\newblock URL: \url{https://onlinelibrary.wiley.com/doi/abs/10.1002/sta4.376}.

\bibitem[Beare and Moon, 2015]{beare2015nonparametric}
Beare, B.~K. and Moon, J.-M. (2015).
\newblock Nonparametric tests of density ratio ordering.
\newblock {\em Econometric Theory}, 31(3):471--492.

\bibitem[Brodeur et~al., 2020]{brodeur2020methods}
Brodeur, A., Cook, N., and Heyes, A. (2020).
\newblock Methods matter: p-hacking and publication bias in causal analysis in
  economics.
\newblock {\em American Economic Review}, 110(11):3634--60.

\bibitem[Brodeur et~al., 2016a]{brodeur2016data}
Brodeur, A., L\'e, M., Sangnier, M., and Zylberberg, Y. (2016a).
\newblock Replication data for: Star wars: The empirics strike back.
\newblock Nashville, TN: American Economic Association [publisher], 2016. Ann
  Arbor, MI: Inter-university Consortium for Political and Social Research
  [distributor], 2019-10-12.
  \url{https://www.openicpsr.org/openicpsr/project/113633/version/V1/view}
  (last accessed 09/23/2020).

\bibitem[Brodeur et~al., 2016b]{brodeur2016}
Brodeur, A., L\'e, M., Sangnier, M., and Zylberberg, Y. (2016b).
\newblock Star wars: The empirics strike back.
\newblock {\em American Economic Journal: Applied Economics}, 8(1):1--32.

\bibitem[Bruns et~al., 2019]{bruns2019reporting}
Bruns, S.~B., Asanov, I., Bode, R., Dunger, M., Funk, C., Hassan, S.~M.,
  Hauschildt, J., Heinisch, D., Kempa, K., K\"onig, J., Lips, J., Verbeck, M.,
  Wolfsch\"utz, E., and Buenstorf, G. (2019).
\newblock Reporting errors and biases in published empirical findings: Evidence
  from innovation research.
\newblock {\em Research Policy}, 48(9):103796.

\bibitem[Carolan and Tebbs, 2005]{carolan2005}
Carolan, C.~A. and Tebbs, J.~M. (2005).
\newblock {Nonparametric tests for and against likelihood ratio ordering in the
  two-sample problem}.
\newblock {\em Biometrika}, 92(1):159--171.

\bibitem[Cattaneo et~al., 2020]{cattaneo2020simple}
Cattaneo, M.~D., Jansson, M., and Ma, X. (2020).
\newblock Simple local polynomial density estimators.
\newblock {\em Journal of the American Statistical Association},
  115(531):1449--1455.

\bibitem[Cattaneo et~al., 2021]{cattaneo2020rddensity}
Cattaneo, M.~D., Jansson, M., and Ma, X. (2021).
\newblock {\em rddensity: Manipulation Testing Based on Density Discontinuity}.
\newblock R package version 2.2.

\bibitem[Christensen and Miguel, 2018]{christensen2018transparency}
Christensen, G. and Miguel, E. (2018).
\newblock Transparency, reproducibility, and the credibility of economics
  research.
\newblock {\em Journal of Economic Literature}, 56(3):920--80.

\bibitem[Cox and Shi, 2020]{cox2020simple}
Cox, G. and Shi, X. (2020).
\newblock Simple adaptive size-exact testing for full-vector and subvector
  inference in moment inequality models.
\newblock {\em arXiv:1907.06317v2}.

\bibitem[de~Winter and Dodou, 2015]{dewinter2015surge}
de~Winter, J.~C. and Dodou, D. (2015).
\newblock A surge of \textit{p}-values between 0.041 and 0.049 in recent
  decades (but negative results are increasing rapidly too).
\newblock {\em PeerJ}, 3:e733.

\bibitem[Elliott et~al., 2020]{elliott2020detecting}
Elliott, G., Kudrin, N., and W\"uthrich, K. (2020).
\newblock Detecting p-hacking.
\newblock {\em arXiv:1906.06711v3}.

\bibitem[Fang, 2019]{fang2019refinements}
Fang, Z. (2019).
\newblock Refinements of the kiefer-wolfowitz theorem and a test of concavity.
\newblock {\em Electron. J. Statist.}, 13(2):4596--4645.

\bibitem[Gerber and Malhotra, 2008]{gerber2008do}
Gerber, A. and Malhotra, N. (2008).
\newblock Do statistical reporting standards affect what is published?
  publication bias in two leading political science journals.
\newblock {\em Quarterly Journal of Political Science}, 3(3):313--326.

\bibitem[Head et~al., 2015]{head2015extent}
Head, M.~L., Holman, L., Lanfear, R., Kahn, A.~T., and Jennions, M.~D. (2015).
\newblock The extent and consequences of p-hacking in science.
\newblock {\em PLoS biology}, 13(3):e1002106.

\bibitem[Head et~al., 2016]{head2016data}
Head, M.~L., Holman, L., Lanfear, R., Kahn, A.~T., and Jennions, M.~D. (2016).
\newblock Data from: {T}he extent and consequences of p-hacking in science.
\newblock Dryad, Dataset.
  \url{https://datadryad.org/resource/doi:10.5061/dryad.79d43} (last accessed
  09/29/2020).

\bibitem[Hung et~al., 1997]{hung1997}
Hung, H. M.~J., O'Neill, R.~T., Bauer, P., and Kohne, K. (1997).
\newblock The behavior of the p-value when the alternative hypothesis is true.
\newblock {\em Biometrics}, 53(1):11--22.

\bibitem[Karunamuni and Alberts, 2005]{karunamuni2005boundary}
Karunamuni, R. and Alberts, T. (2005).
\newblock On boundary correction in kernel density estimation.
\newblock {\em Statistical Methodology}, 2(3):191 -- 212.

\bibitem[Kulikov and Lopuha{\"{a}}, 2008]{kulikov2008distribution}
Kulikov, V.~N. and Lopuha{\"{a}}, H.~P. (2008).
\newblock Distribution of global measures of deviation between the empirical
  distribution function and its concave majorant.
\newblock {\em Journal of Theoretical Probability}, 21(2):356--377.

\bibitem[Leggett et~al., 2013]{leggett2013life}
Leggett, N.~C., Thomas, N.~A., Loetscher, T., and Nicholls, M. E.~R. (2013).
\newblock The life of p: ``just significant'' results are on the rise.
\newblock {\em The Quarterly Journal of Experimental Psychology},
  66(12):2303--2309.

\bibitem[Masicampo and Lalande, 2012]{masicampo2012peculiar}
Masicampo, E.~J. and Lalande, D.~R. (2012).
\newblock A peculiar prevalence of p values just below .05.
\newblock {\em The Quarterly Journal of Experimental Psychology},
  65(11):2271--2279.

\bibitem[McCrary, 2008]{mccrary2008manipulation}
McCrary, J. (2008).
\newblock Manipulation of the running variable in the regression discontinuity
  design: A density test.
\newblock {\em Journal of Econometrics}, 142(2):698--714.

\bibitem[{R Core Team}, 2020]{R20}
{R Core Team} (2020).
\newblock {\em R: A Language and Environment for Statistical Computing}.
\newblock R Foundation for Statistical Computing, Vienna, Austria.

\bibitem[Simonsohn et~al., 2014]{simonsohn2014p}
Simonsohn, U., Nelson, L.~D., and Simmons, J.~P. (2014).
\newblock P-curve: a key to the file-drawer.
\newblock {\em Journal of Experimental Psychology: General}, 143(2):534--547.

\bibitem[Simonsohn et~al., 2015]{simonsohn2015better}
Simonsohn, U., Simmons, J.~P., and Nelson, L.~D. (2015).
\newblock Better p-curves: {M}aking p-curve analysis more robust to errors,
  fraud, and ambitious p-hacking, a reply to {U}lrich and {M}iller (2015).
\newblock {\em Journal of Experimental Psychology: General}, 144(6):1146--1152.

\bibitem[Snyder and Zhuo, 2018]{snyder2018sniff}
Snyder, C. and Zhuo, R. (2018).
\newblock Sniff tests in economics: Aggregate distribution of their probability
  values and implications for publication bias.
\newblock NBER WP 25058.

\bibitem[StataCorp., 2019]{stata2019}
StataCorp. (2019).
\newblock {\em Stata Statistical Software: Release 16.}
\newblock College Station, TX.

\bibitem[Ulrich and Miller, 2015]{ulrich2015p}
Ulrich, R. and Miller, J. (2015).
\newblock p-hacking by post hoc selection with multiple opportunities:
  Detectability by skewness test?: Comment on {S}imonsohn, {N}elson, and
  {S}immons (2014).
\newblock {\em Journal of Experimental Psychology: General}, 144:1137--1145.

\bibitem[Ulrich and Miller, 2018]{ulrich2018some}
Ulrich, R. and Miller, J. (2018).
\newblock Some properties of p-curves, with an application to gradual
  publication bias.
\newblock {\em Psychological Methods}, 23(3):546–560.

\bibitem[Vivalt, 2019]{vivalt2019specification}
Vivalt, E. (2019).
\newblock Specification searching and significance inflation across time,
  methods and disciplines.
\newblock {\em Oxford Bulletin of Economics and Statistics}, 81(4):797--816.

\end{thebibliography}

\begin{appendix}

\section{Additional details Section \ref{sec: tests t-tests}}
\label{app: additional details}

\subsection{Bounds on proportions and their differences}\label{app: details cox and shi test 1}
The bounds on the proportions and their differences implied by hypothesis \eqref{eq:testing_additional_constraints} are not sharp in general. Here we derive sharp bounds by directly extremizing the proportions and their differences.

For the one-sided $t$-tests, the population proportion, $\pi_j$, can be written as
\footnotesize
\begin{eqnarray*}
\pi_j =\int_{x_{j-1}}^{x_j}g_1(p)dp
&=&\int_{x_{j-1}}^{x_j}\int_{{[0,\infty)}}e^{-h^2/2}e^{hcv_1(p)}d\Pi(h)dp\\
&=&\int_{{[0,\infty)}}\left( \int_{x_{j-1}}^{x_j}e^{-h^2/2}e^{hcv_1(p)}dp\right)d\Pi(h)\\
&=& \int_{{[0,\infty)}}\left( \int_{cv_1(x_{j})}^{cv_1(x_{j-1})}\phi(t-h)dt\right)d\Pi(h)\\
&=& \int_{{[0,\infty)}}\lambda_{1,j}(cv_1, h)d\Pi(h),
\end{eqnarray*}
\normalsize
where $\lambda_{1,j}(cv, h) := \Phi(cv(x_{j-1})-h) - \Phi(cv(x_{j})-h)$. For the two-sided $t$-tests,
$
\pi_j =\int_{x_{j-1}}^{x_j}g_2(p)dp= \int_{\mathbb{R}}\lambda_{2,j}(cv_2, h)d\Pi(h),
$
where $\lambda_{2,j}(cv, h) :=\lambda_{1,j}(cv,h)+\lambda_{1,j}(cv, -h)$.

Since $\lambda_{1,j}(cv_1,h)$, as a function of $h$, attains its maximum at $h_j^\ast=\frac{cv_1(x_{j-1})+cv_1(x_j)}{2}$, for the one-sided $t$-tests $\pi_j \le 2\Phi \left(\frac{cv_1(x_{j-1})-cv_1(x_j)}{2}\right)-1:=\vartheta^{(0)}_{1,j}$. In case of the two-sided $t$-tests, the bound, $\vartheta^{(0)}_{2,j}:=\max_{h\in \mathbb{R}}\lambda_{2,j}(cv_2,h)$, can be calculated numerically.

For the bounds on the k$^{th}$ differences of $\pi$'s, note that, for  $j=1,\dots, J-k$, $\Delta^{k}_j = \sum_{i=0}^{k}(-1)^{i}{\binom{k}{i}}\pi_{k+j-i}$ 
and therefore
\begin{equation*}
    |\Delta^{k}_j|\le \vartheta^{(k)}_{s,j}:=\max_{h\in \mathcal{H}_{(s)}}\left\{\sum_{i=0}^{k}(-1)^{i+k}{\binom{k}{i}}\lambda_{s,k+j-i}(cv_s, h)\right\},\quad j=1,\dots, J-k,
\end{equation*}
where $\mathcal{H}_{(1)} = [0, \infty)$, $\mathcal{H}_{(2)}=\mathbb{R}$, and $s=1$ and $s=2$ for the one- and two-sided $t$-tests, respectively. These bounds can be computed numerically.

\subsection{Null hypothesis}\label{app: details cox and shi test 2}

The null hypothesis formulated in terms of the proportions is
\begin{equation}
H_0:~ 0\le(-1)^k\Delta^k\le \boldsymbol{\vartheta}_s^{(k)},\: \sum_{j=1}^{J}\pi_j=1,~~ \text{for all}~ k=0,\dots,K, \label{eq:H0 Delta}
\end{equation}
where $\Delta^k$ is a $(J-k)\times 1$ vector of k$^{th}$ differences of $\pi$'s, $\Delta^0=\boldsymbol{\pi}$, 
$\boldsymbol{\vartheta_{s}^{(k)}}:=(\vartheta_{s,1}^{(k)}, \dots, {\vartheta}_{s,J-k}^{(k)})'$ is the vector of upper bounds on $|\Delta^k|$ (cf.\ Appendix \ref{app: details cox and shi test 1}), $s=1$ for one-sided tests, and $s=2$ for two-sided tests. The inequalities in \eqref{eq:H0 Delta} are interpreted element-wise.

Let $D_m$ be $(m-1)\times m$ differencing matrix of the following form:
\footnotesize
$$D_{m} := \begin{pmatrix}-1&1&0&\dots&0&0\\
\vdots&\vdots&\vdots&\ddots&\vdots&\vdots\\
0&0&0&\dots&-1&1
 \end{pmatrix}.$$
 \normalsize
In addition, define the $J\times 1$ vector $e_J := (0,\dots, 1)'$, $(J-1)\times 1$ vector $i_{J-1}:=(1,\dots,1)'$, and matrix $F := [-I_{J-1}, i_{J-1}]'$.
Using this notation, we can write $(-1)^{k}\Delta^k = D^k\boldsymbol{\pi},~ k=1,\dots, K$,
where $D^k := (-1)^k D_{J-k+1}\times\cdots\times D_{J}$. Note that the restrictions under the null are equivalent to $\mathcal{D}_K\boldsymbol{\pi}\ge c \text{ and }\boldsymbol{\pi} = e_{J} - F\boldsymbol{\pi}_{-J}$, where $\mathcal{D}_K = [-1, 1]'\otimes[I_{J}, D^{1'}, \dots, D^{K'}]'$ and 
 $c = [\boldsymbol{\vartheta_{s}^{(0)}}',\dots, \boldsymbol{\vartheta_{s}^{(K)}}', 0'_{(K+1)(J-K/2)\times 1}]'$. The symbol $\otimes$ denotes the Kronecker product. We can thus express the null hypothesis \eqref{eq:H0 Delta} as $H_0:~ A\boldsymbol{\pi}_{-J}\le b$, where $A :=\mathcal{D}_KF$ and $b := \mathcal{D}_Ke_{J}-c$.

When testing on a subinterval $(0, \alpha]$, the bounds need to be re-scaled. We use a consistent (under the null) estimator of $G(\alpha)$ to re-scale the bounds. In particular, we use bounds $\vartheta^{(k)}_{s,j}=\vartheta^{(k)}_{s,j}/\hat{G}(\alpha)$, where $\hat{G}(\alpha)$ is the fraction of $p$-values below $\alpha$.

\section{Proofs}

\subsection{Proof of Lemma \ref{lem:verification_sufficient_condition}}
\label{app:lem:verification_sufficient_condition}

Note that for claim (i) $\{cv(p): p\in (0,1)\}=\mathbb{R}$ and for claims (ii) and (iii) $\{cv(p): p\in (0,1)\}=(0,\infty)$.

\medskip

\noindent \textit{Claim (i):} In this case $f(x) = \phi(x)$ and $f_h(x) = \phi(x-h)$. It follows that, for all $h\ge 0$, $f'_h(x)f(x) - f'(x)f_h(x)=h\phi(x)\phi(x-h) \ge 0.$

\medskip

\noindent \textit{Claim (ii):} In this case $f(x) = 2\phi(x)$ and $f_h(x) = \phi(x-h)+\phi(x+h)$, where $x\ge 0$. After taking derivatives and collecting terms we get
\begin{equation*}
f'_h(x)f(x) - f'(x)f_h(x)=2\phi(x)h(\phi(x-h) - \phi(x+h))=2\phi(x)\phi(x+h)h(e^{2xh}-1) \ge 0,
\end{equation*}
because $h(e^{2xh}-1)\ge 0$ for any $h$.

\medskip

\noindent \textit{Claim (iii):} In this case $f(x) := f(x;d) =  \frac{1}{2^{d/2}\Gamma(d/2)}x^{d/2-1}e^{-x/2}$ and
$f_h(x) =\sum_{j=0}^{\infty}\frac{e^{-h/2}(h/2)^j}{j!}f(x; d+2j)$,
where $x>0$. Note that $f'(x;d) = f(x;d) \left((d-2)x^{-1}-1\right)/2$. After taking derivatives and collecting terms we get
\footnotesize
\begin{eqnarray*}
f'_h(x)f(x) - f'(x)f_h(x)&=&\sum_{j=0}^{\infty}\frac{e^{-h/2}(h/2)^j}{2j!}f(x; d+2j)f(x;d)\left[((d+2j-2)x^{-1}-1) - ((d-2)x^{-1}-1)\right]\\
 &=&\sum_{j=0}^{\infty}\frac{e^{-h/2}(h/2)^j}{j!}f(x; d+2j)f(x;d)jx^{-1} \ge 0,
\end{eqnarray*}
\normalsize
since every term in the last sum is non-negative. \qed

\subsection{Proof of Theorem \ref{thm:monotonicity}}
\label{app:thm:monotonicity}
Recall that $\beta(p,h)=1-F_h\left(cv(p)\right)$, where $cv(p)=F^{-1}(1-p)$. Under Assumption \ref{ass:regularity}, 
\begin{eqnarray*}
\frac{\partial^2 \beta(p,h)}{\partial p^2}&=&\frac{f'_{h}(cv(p))cv'(p)f(cv(p))-f'(cv(p))cv'(p)f_{h}(cv(p))}{f(cv(p))^2}\\
&=&\frac{cv'(p)}{f(cv(p))^2}\left[ f'_{h}(cv(p))f(cv(p))-f'(cv(p))f_{h}(cv(p))\right] .
\end{eqnarray*}
Non-increasingness of $g$ now follows by Assumption \ref{ass:sufficient_condition} and because $cv'(p)/f(cv(p))^2\le 0$. Continuous differentiability is implied by Assumption \ref{ass:regularity}. \qed

\subsection{Proofs of Theorems \ref{thm: complete monotonicity} and \ref{thm: bounds}}
\label{app: proof bounds}

Note that the $p$-curves for the one-sided and two-sided $t$-tests are given by
\begin{eqnarray}
g_1(p) &=& \int_{[0,\infty)}\Psi(cv_1(p),h)\exp\{-h^2/2\}d\Pi(h),\label{eq:g_1_K}\\
g_2(p) &=& \frac{1}{2}\int_{\mathbb{R}}(\Psi(cv_2(p),h)+\Psi(cv_2(p),-h))\exp\{-h^2/2\}d\Pi(h)\label{eq:g_2_K}
\end{eqnarray}
where $\Psi(x,y) := \exp\{xy\}$. We start by proving an auxiliary lemma about $\Psi(x,y)$.

\begin{lemma}\label{lem:auxiliary}
For $k\ge1$, the k$^{th}$ derivative of $\Psi(cv_s(p),h)$ is
\begin{equation*}
\Psi^{(k)}(cv_s(p),h) = (-1)^{k}\frac{h \sum_{j=0}^{k-1}A^{k}_j(cv_s(p))[cv_s(p) + h]^{j}}{s^k(\phi(cv_s(p)))^k}\Psi(cv_s(p), h),
\end{equation*}
where coefficients $A^k_j(cv_s(p))$ are polynomials in $cv_s(p)$ with non-negative coefficients and $s=1$ for one-sided and $s=2$ for two-sided $t$-tests.
\end{lemma}
\begin{proof}
By direct computation, the first derivative of $\Psi(cv_s(p),h)$ with respect to $p$ is
$\Psi^{(1)}(cv_s(p),h) = -\frac{h}{s\phi(cv_s(p))}\Psi(cv_s(p), h).$
We use induction to derive the k$^{th}$ derivative of $\Psi(cv_s(p), h)$. Suppose that for $k>1$
\begin{equation*}
\Psi^{(k)}(cv_s(p),h) = (-1)^{k}\frac{h \sum_{j=0}^{k-1}A^{k}_j(cv_s(p))[cv_s(p) + h]^{j}}{s^k(\phi(cv_s(p)))^k}\Psi(cv_s(p), h),
\end{equation*}
where coefficients $A^k_j(cv_s(p))$ are polynomials in $cv_s(p)$ with non-negative coefficients. Define $B^k_0 = (k-1)cv_s(p)A_0^k(cv_s(p))$, $B^{k}_{j} = (k-1)cv_s(p)A^{k}_{j}(cv_s(p))+A^{k}_{j-1}(cv_s(p))$ for $j=1,\dots,k-1$, and $B^{k}_{k} = A^{k}_{k-1}(cv_s(p))$; $C^k_j = \partial A^{k}_{j}(cv_s(p))/\partial cv_s(p) + (j+1)A^k_{j+1}(cv_s(p))$ for $j=0,\dots,k-2$, $C^k_{k-1} = \partial A^{k}_{k-1}(cv_s(p))/\partial cv_s(p)$, and $C^{k}_{k}=0$. Now differentiate $\Psi^{(k)}(cv_s(p),h)$ with respect to $p$ to get

\footnotesize
\begin{eqnarray*}
\Psi^{(k+1)}(cv_s(p),h)&=&(-1)^{k+1}\frac{h^2 \sum_{j=0}^{k-1}A^{k}_j(cv_s(p))[cv_s(p) + h]^{j}}{s^{k+1}(\phi(cv_s(p)))^{k+1}}\Psi(cv_s(p), h)\\
&&+(-1)^{k+1}\frac{(hcv_s(p)k) \sum_{j=0}^{k-1}A^{k}_j(cv_s(p))[cv_s(p) + h]^{j}}{s^{k+1}(\phi(cv_s(p)))^{k+1}}\Psi(cv_s(p), h)\\
&&+(-1)^{k+1}\frac{h \sum_{j=0}^{k-1}(\partial A^{k}_j(cv_s(p))/\partial cv_s(p))[cv_s(p) + h]^{j}}{s^{k+1}(\phi(cv_s(p)))^{k+1}}\Psi(cv_s(p), h)\\
&&+(-1)^{k+1}\frac{h \sum_{j=1}^{k-1} j A^{k}_j(cv_s(p))[cv_s(p) + h]^{j-1}}{s^{k+1}(\phi(cv_s(p)))^{k+1}}\Psi(cv_s(p), h)\\
&=& (-1)^{k+1}\frac{\Psi(cv_s(p), h)}{s^{k+1}(\phi(cv_s(p)))^{k+1}}\left\{h\sum_{j=0}^{k}(B^k_j + C^k_j)[cv_s(p) + h]^{j}\right\}.
\end{eqnarray*}
\normalsize
Since $A^{k}_j(cv_s(p)), j=0,\dots, k-1$ are polynomials with non-negative coefficients, $B^k_j$ and $C^k_j$ are also polynomials with non-negative coefficients for every $j=0,\dots,k$. It follows that
$$\Psi^{(k+1)}(cv_s(p),h) = (-1)^{k+1}\frac{h \sum_{j=0}^{k}A^{k+1}_j(cv_s(p))[cv_s(p) + h]^{j}}{s^{k+1}(\phi(cv_s(p)))^{k+1}}\Psi(cv_s(p), h),$$
where $A^{k+1}_j(cv_s(p)) = B^k_j+C^k_j, j=0,\dots, k$. This completes the induction step.
\end{proof}

Using Lemma \ref{lem:auxiliary}, we now proof Theorem \ref{thm: complete monotonicity} and Theorem \ref{thm: bounds}.

\begin{proof}[Proof of Theorem \ref{thm: complete monotonicity}]
Lemma \ref{lem:auxiliary} and equations \eqref{eq:g_1_K}--\eqref{eq:g_2_K} directly imply that $0\le (-1)^{k}g_1^{(k)}(p)$, for $p\in (0,1/2]$ and $0\le (-1)^{k}g_2^{(k)}(p)$, for $p\in (0,1)$ for $k=1,2,\dots$. The result for the two-sided case follows from the fact that $h\{[cv_2(p)+h]^j\Psi(cv_2(p),h) -[cv_2(p)-h]^j\Psi(cv_2(p),-h)\}\ge0$ for every $j\in \mathbb{N}$ and every $h\in \mathbb{R}$.
\end{proof}

\begin{proof}[Proof of Theorem \ref{thm: bounds}]
Consider first the one-sided $t$-test. Lemma \ref{lem:auxiliary} implies that
\begin{equation*}
    (-1)^{k}g_1^{(k)}(p)\le \mathcal{B}_{1}^{(k)}(p):=\max_{h\ge 0}\left\{ |\Psi^{(k)}(cv_1(p),h)|\exp\{-h^2/2\}\right\},
\end{equation*}
where the inequality holds for every $p\in(0,1)$ and the maximum is finite for every $p\in(0,1)$ since $|\Psi^{(k)}(cv_1(p),h)|\exp\{-h^2/2\}$ is finite for every $h\ge0$ and converges to zero as $h$ goes to infinity. For the upper bound on $g_1(p)$, note that for $p\in (0, 1/2]$, $\max_{h\ge 0}\left\{ |\Psi(cv_1(p),h)|\exp\{-h^2/2\}\right\}=\Psi(cv_1(p),cv_1(p))\exp\{-cv_1^2(p)/2\}=\exp\{cv_1^2(p)/2\}$. For $p>1/2$ and $h\ge 0$, $hcv_1(p)-cv_1^2(p)/2<0$ and hence $g_1(p)\le 1$. 

For two-sided tests, by the above arguments and symmetry, we have
\begin{equation*}
    (-1)^{k}g_2^{(k)}(p)\le \mathcal{B}_2^{(k)}(p):=\max_{h\in \mathbb{R}}\left\{ |\Psi^{(k)}(cv_2(p),h)+\Psi^{(k)}(cv_2(p),-h)|\exp\{-h^2/2\}/2\right\},
\end{equation*}
where the upper bound is finite for every $p\in (0,1)$. 

For the upper bound on $g_2(p)$, one can show that for $p\ge 2(1-\Phi(1))$, the first-order condition for maximizing $|\Psi(cv_2(p),h)+\Psi(cv_2(p),-h)|\exp\{-h^2/2\}/2$ has only one solution, $h_o=0$. By checking second-order conditions we can verify that $0$ is the maximum. For $p< 2(1-\Phi(1))$, $0$ becomes local minimum, and there are two additional non-zero symmetric solutions to the first-order condition that satisfy the second-order condition for a maximum and result in identical values of the objective function.

\end{proof}

\end{appendix}

\end{document}